\PassOptionsToPackage{dvipsnames}{xcolor}

\newcommand\SHORTERICALP[1]{}

\documentclass{article}
\usepackage{mathptmx}
\usepackage[utf8]{inputenc}
\usepackage[T1]{fontenc}

\usepackage[final]{microtype}
\usepackage[english]{babel}
\usepackage{csquotes}
\usepackage{amsmath}
\usepackage{amsthm}
\usepackage{amsfonts}
\usepackage{amssymb}
\usepackage{commath}
\usepackage{bm}
\usepackage{xfrac}
\usepackage{mathtools}
\usepackage{authblk}
\usepackage{url}
\usepackage[pdfencoding=auto,psdextra,hidelinks]{hyperref}
\usepackage{bookmark}
\usepackage{float}
\usepackage{booktabs}
\usepackage{subcaption}
\usepackage{tikz}
\usetikzlibrary{calc,positioning,shapes,arrows}
\usepackage[dvipsnames]{xcolor}
\usepackage[nottoc, notlof, notlot]{tocbibind}

\newcommand\manonclass{\linearderivlength^\bullet}
\newcommand\MANONlim{ELim}

\newcommand\manonclasslim{\bar{\manonclass}}

\newcommand{\signb}[1]{\bar{\fonction{cond}}(#1)}

\newcommand\Encode{\textit{Encode}}
\newcommand\Decode{\textit{Decode}}
\newcommand{\composition}{\textit{composition}}
\newcommand{\linearLODE}{\textit{linear length ODE}}
\newcommand{\robODE}{\textit{robust ODE}}
\newcommand{\contClasslim}{\bar{\mathbb{RCD}}}
\newcommand{\contClass}{\mathbb{RCD}}
\newcommand{\contClasslight}{\mathbb{RCD_{*}}}

\newcommand\spaceclasstanhlim{\bar{\mathbb{RLD}^\circ}}

\newcommand\manonclasslighttanh{\linearderivlength^\circ}

\newcommand\manonclasslighttanhlim{\bar{\manonclasslighttanh}}

\newcommand\EncodeMul{\textit{EncodeMul}}
\newcommand\DP{\operatorname{DP}}

\newcommand\send{\operatorname{send}}
\newcommand\smooth{\mathcal{C}}
\newcommand\sendtanh{\operatorname{\smooth-send}}
\newcommand\sendsymbol{\mapsto}

\newcommand\TTtanh{\operatorname{\smooth-if}}

\newcommand\relu{\operatorname{ReLU}}

\newcommand\sig{\operatorname{\mathfrak{s}}}
\newcommand\sigtanh{\operatorname{\mathfrak{\smooth-s}}}
\newcommand\relutanh{\operatorname{ReLU\mathfrak{-s}}}

\usepackage{versions}

\newcommand\vectorl[1]{{\mathbf#1}}

\newcommand\vn{\vectorl{n}}
\newcommand\vx{\vectorl{x}}

\newcommand\R{\mathbb{R}}

\newcommand\N{\mathbb{N}}
\newcommand\Z{\mathbb{Z}}

\newcommand{\fonction}[1]{\textrm{#1}}

\newcommand\projection[2]{\mathbf{\pi}_{#1}^{#2}}

\newcommand\plus{\mathbf{+}}
\newcommand\minus{\mathbf{-}}

\newcommand{\sign}[1]{\fonction{sg}(#1)}

\newcommand{\tu}[1]{\mathbf{#1}}

\newcommand{\cp}[1]{\mathbf{#1}}
\newcommand{\Ptime}{\cp{PTIME}} 

\newcommand{\FPtime}{\cp{FPTIME}}

\newcommand{\Pspace}{\cp{PSPACE}}
\newcommand{\NPspace}{\cp{NPSPACE}}
\newcommand{\FPspace}{\cp{FPSPACE}}

\newcommand{\NLOGSPACE}{\ensuremath{\operatorname{NLOGSPACE}}}

\newcommand{\NSPACE}{\ensuremath{\operatorname{NSPACE}}}
\newcommand{\sSPACE}{\ensuremath{\operatorname{SPACE}}}

\renewcommand\bar[1]{\overline{#1}}

\usepackage{color}

\newcommand{\dderiv}[2]{\frac{\delta #1}{\delta #2}}
\newcommand{\dderivL}[1]{\frac{\delta #1}{\delta \lengt}}
\newcommand{\dderivl}[1]{\frac{\partial #1}{\partial \ell}}

\newcommand{\lengt}{\mathcal{L}}
\newcommand\lengthnotation{\ell}
\newcommand{\length}[1]{\mathrm{\lengthnotation}(#1)}

\newcommand{\linearderivlength}{\mathbb{LDL}}

\usepackage{todonotes}

\newcommand\encodagemot{\gamma_{word}}
\newcommand\base{4}
\newcommand\symboleun{1}
\newcommand\symboledeux{3}

\newcommand\Image{\mathcal{I}}

\newcommand{\Next}{\textit{Next}}

\newtheorem{theorem}{Theorem}
\newtheorem{proposition}{Proposition}
\newtheorem{lemma}{Lemma}
\newtheorem{corollary}{Corollary}

\newtheorem{definition}{Definition}
\newtheorem{remark}{NB}

\newtheorem{example}{Example}

\title{The complexity of computing 
in continuous time: space complexity  is precision}
\author{Manon Blanc \\ 
	manon.blanc@lix.polytechnique.fr \\
	École Polytechnique, LIX, Palaiseau, France, Université Paris-Saclay, LISN, Orsay, France
	\and Olivier Bournez \\
	olivier.bournez@lix.polytechnique.fr\\
	École Polytechnique, LIX, Palaiseau, France}
\date{ }

\begin{document}
	
	\maketitle

\begin{abstract}
Models of computations over the integers are equivalent from a computability and complexity theory point of view by the Church-Turing thesis. It is not possible to unify discrete-time models over the reals. 
The situation is unclear but simpler for continuous-time models, as there is a unifying mathematical model, provided by ordinary differential
equations (ODEs). Each model corresponds to a particular class of  ODEs. 
For example, the General Purpose Analog Computer model of Claude Shannon, introduced as a mathematical model of analogue machines (Differential Analyzers),  is known to correspond to polynomial ODEs.  However, the question of a robust complexity theory for such models and its relations to classical (discrete) computation theory is an old problem.
There was some recent significant progress:  it has been proved that (classical) time complexity corresponds to the length of the involved curves, i.e. to the length of the solutions of the corresponding polynomial ODEs. The question of whether there is a simple and robust way to measure space complexity remains. We argue that space complexity corresponds to precision and conversely.

Concretely, we propose and prove an algebraic characterisation of $\FPspace$, using 
\emph{continuous} ODEs. 
Recent papers proposed algebraic 
characterisations of polynomial-time and -space complexity classes over the reals, but with a discrete-time: those algebras rely on discrete ODE schemes.
Here, we use classical (continuous) ODEs, with the classic definition of derivation and hence with the more natural context of continuous-time associated with ODEs. We characterise both the case of polynomial space functions over the integers and the reals. 
This is done by proving two inclusions. The first is obtained using some original polynomial space method for solving ODEs. For the other, 
we prove that Turing machines, with a proper representation of real numbers, can be simulated by continuous ODEs and not just discrete ODEs. 
A major consequence is that the associated space complexity is provably related to the numerical stability of involved schemas and the associated required precision. We obtain that a problem can be solved in polynomial space if and only if it can be simulated by some numerically stable ODE, using a polynomial precision. 

\end{abstract}

\section{Introduction}

Recently, there has been a renewed interest in models of computations over the reals and their associated complexity classes. The fact that these models appear in complexity issues of deep learning models (a.k.a. neural networks) partially explains it. For example, various problems, such as the training of fully connected neural networks, have been proved to be a $\exists \R$-complete problem \cite{bertschinger10training}. 
Complexity classes like $\cp{FIXP}$ were introduced to discuss the complexity of continuous functions' fixed points in various contexts, such as game theory \cite{etessami2010complexity}. These classes and statements are related to discrete-time models of computation over the reals. 

For discrete-time models of computations over the reals, the most famous approaches are computable analysis, based on the Turing machine model in \cite{Tur36} and \cite{Wei00} and algebraic models such as the Blum Shub Smale (BSS) model of computation \cite{BSS89,BCSS98}. The class $\exists \mathbb{R}$ corresponds to the (constant-free, equivalently uniform) non-deterministic time of the BSS model of computation. An extensive list of decision problems was proved recently to be in this class.
Both models were tailored for very different applications and it is well-known we cannot unify existing models with the equivalent of a Church-Turing thesis. For example, computable functions in a computable analysis model need to be continuous, while the BSS model intends to consider functions and problems over the polynomials that are not. It is also explained by the fact that some models have not been introduced with the idea of corresponding to actual physical machines but also to discuss abstract complexity (lower and upper bounds) for associated problems. 

Among models of computation over the reals, we can also distinguish continuous-time models. This includes models of old, first-ever built computers, such as the Differential Analysers \cite{ulmann2020analog}. A famous mathematical model of such machines is the General Purpose Analog Computer model of Claude Shannon \cite{Sha41}. It covers many historical machines and today's analogue devices \cite{LivreAnalogcomputing,veritasum} too. It also includes various recent approaches and models from deep learning such as Neural ODEs \cite{chen2018neural,kidger2022neural} with many variants. In the context of continuous-time, the situation is clearer than with discrete-time models, as there is a unifying way to describe these models, provided by ODEs. Each model corresponds to a particular class of  ODEs.  For example, the GPAC corresponds to polynomial ODEs \cite{GraCos03}, and Neural ODEs are made by selecting
the best solution among a parameterised class of ODEs: see, e.g. \cite{kidger2022neural}. 

Even if particular classes of ODES can describe such models, defining a robust and well-defined computation theory for continuous-time computations is not an easy problem: see \cite{bournez2021survey} for the most recent survey. In short, the problem with time complexity is that considering the time variable as a measure of time is not robust: a curve can always be re-parameterised using a change of variable. The problem with space complexity is similar: reparameterisation corresponds to a change of time variable, but also of space-variable, introducing space and time contractions: See e.g. \cite{bournez2021survey,JournalACM2017}. Furthermore, many problems for simple dynamical systems are known to be undecidable, hence forbid $\Pspace$-completeness see \cite{dsg06a} and \cite{GracaZhongHandbook}. 

There was a recent breakthrough in \cite{ICALP2016vantardise,JournalACM2017}, where the authors relate time with the length of the solution curve of an ODE. It holds for polynomial ODEs, and their projections, known to cover a very wide class of functions, including all common functions or functions that can be built from them \cite{TheseDaniel}.  As the length of a curve is an invariant, this solves the issue of a possible change of variable. Furthermore, the authors prove that for polynomial ODEs, this is polynomially related 
to the time required to solve an ODE, hence providing a robust notion of time for ODEs. 
These statements and underlying constructions, which allow the simulation of Turing machines, 
led to solving several open problems: the existence of a universal ODE \cite{ICALP2017}, the proof of the  Turing-completeness of chemical reactions \cite{CMSB17vantardise}, or statements about the hardness of several dynamical systems problems \cite{GracaZhongHandbook}. 

The question of whether we can give a simple equivalent defining \emph{space-complexity} remains. We argue here that space complexity is polynomially related and conversely to the numerical stability of ODEs and their associated precision. We prove that a problem can be solved in polynomial space
iff it can be simulated by some numerically stable ODE, using a polynomial precision.  
We prove this holds both for classical complexity over the discrete (functions over the integers) and also for space complexity for real functions in the model of computable analysis. 

\begin{remark}
	In the literature, there are two possible definitions for $\FPspace$, according to whether functions with non-polynomial size values are allowed or not. In this article, when we talk about $\FPspace$, we always assume the outputs remain of polynomial size.
	Otherwise, the class is not closed by composition: the issue is about not counting the output 
	as part of the total space used. Given $f$ computable in polynomial space and $g$ in logarithmic space, $f \circ g$ (and $g \circ f$) is computable in polynomial space. But, if exponential size output is allowed, this is not true: if we assumed only $f$ and $g$ to be computable in polynomial space since the first might give an output of exponential size.
\end{remark}

These questions of providing characterisations of classical complexity using ODEs can also be seen from the so-called ``implicit complexity'' point of view. 
Having "simple" characterisations of computability and complexity classes is useful
for various fundamental and applied science fields. We are interested here in "algebraic"
characterisations of those classes: we want to define them as the smallest set $ [f_1, \dots, f_k; o_1, \dots, o_l] $ where the $f_i$'s are functions, closed under the operators $o_i$'s. For example, the set of computable functions over the integers is well-known to be: 
$$[0, 1, \pi^i_k; \textit{composition, minimisation, primitive recursion}].$$
 Implicit complexity aims at giving similar algebras for classes of complexity theory: a reference survey is \cite{Clo95,clote2013boolean}.
The main benefit is to avoid the use of the framework of Turing machines, which is rather 
heavy and not necessarily well-known outside fundamental computer science.
Several characterisations for $\Ptime$ over the integers were proposed. 
The first is due to Cobham in \cite{Cob65}, but relies on explicit ah-hoc bounds. 
Other approaches have then been proposed, see surveys
\cite{Clo95,clote2013boolean}.  Recently, Bournez and Durand in \cite{MFCSJournal}
suggested an algebra using the so-called "linear-length" discrete ODEs. Instead of 
having explicit bounds, the linearity of the involved discrete ODE guarantees polynomial time complexity.

Using a similar approach, Blanc and Bournez in \cite{BlancBournezMCU22vantardise} and in \cite{BlancBournezMFCS2023vantardise} extended the constructions to a characterisation of $\Ptime$ for function over 
the reals. The latter extended the result to $\Pspace$, defining \emph{robust} ODEs as in Definition \ref{def:roblinear lengt ODE}. However, those models rely
on discrete ODEs, which are discrete-time and less natural than continuous ODEs.
We review all those results in Section \ref{sectionPrevioussoa}. 

This paper can be related to \cite{ICALP2016vantardise,JournalACM2017}: the authors of these articles provided a characterisation 
of $\Ptime$ with \emph{continuous} ODEs,  establishing that time complexity corresponds to the length of the involved curve, i.e. the motto time complexity = length. Here, we get a motto of the form \textbf{space complexity = precision}. Some of our constructions have similarities with statements in \cite{csl24}. In the later paper, the authors introduce various robustness concepts and prove they lead to tractability. 
See the references in \cite{csl24} for similar robustness statements. Robustness can also be associated with a dual motivation: the authors of \cite{gracca2023robust} introduced a concept of robust undecidability, while here we want a concept of robustness leading to tractability. 

This is not the first time $\FPspace$ is characterised using continuous ODES. However, the existing characterisation \cite{TheseRiccardo,BGDPRiccardo2022} is obtained with complicated conditions on ODEs, while we have a simpler statement, linking complexity to precision in a simple manner. Notice that the latter approach dealt with polynomial ODEs, while we do not restrict to polynomial ODEs. We obtain our statements by revisiting the approach of the latter papers but working over a compact domain and dealing with error correction more finely. 

Intuitively, this can also be read as being in $\Pspace$ for an ODE is consistent with having an attractor easily discretisable when there is one. We can also define the 
notion of robustness, as the insensitivity to ``small" perturbations.

While discussing all these issues, we propose an algebraic characterisation of $\Pspace$, using \emph{continuous} ODEs with the algebra ($\contClass$ is for Continuous Robust Differential).  Schema \robODE{}  is formally defined in Definition \ref{defRobODE}: 

$$\contClass = [0, 1,\pi_i^k, +, -, \times, \tanh, \cos, \pi, \frac{x}{2}, \frac{x}{3}; \composition, \robODE]$$ 

For a function $\tu f: \R^{d} \to \R^{d'}$ sending every integer $\vn \in \N^{d}$ to the vicinity of some integer of  $\N^{d}$, say at distance less than $1/4$, we write $\DP(f)$ for its discrete part: this is the function from $\N^{d} \to \N^{d'}$ mapping $\vn \in \N^{d}$ to the integer rounding of $\tu f(n)$. For a class $\mathcal{C}$ of such functions, we write $\DP(\mathcal{C})$ for the class of the discrete parts of the functions of $\mathcal{C}$. 

\begin{theorem} \label{th:mainone}
$\DP(\contClass)= \FPspace  \cap \N^{\N}.$
\end{theorem}

We also provide a characterisation of functions over the reals computable in polynomial space. Inspired by \cite{BlancBournezMFCS2023vantardise}, this is obtained by adding a limit schema $\MANONlim$ to $\contClass$. If we consider  $\contClasslim = [0, 1,\pi_i^k, +, -, \times, \tanh, \cos, \pi, \frac{x}{2}, \frac{x}{3}; \composition, \robODE, \MANONlim]$ then:

\begin{theorem}[Generic functions over the reals] \label{th:main:twopSpace} 
	{$\contClasslim \cap \R^{\R}  =  \FPspace \cap \R^{\R} $} \\
	More generally: \quad \quad \quad  \quad  \hspace{0.1cm}
	   {$\contClasslim \cap \R^{\N^{d} \times \R^{d'}}  = \FPspace \cap \R^{\N^{d} \times \R^{d}} $.} 
\end{theorem}

This article is organised as follows. 
In Section \ref{sectionDynSys}, we recall the concept of dynamical systems and discuss some associated complexity issues. We introduce the concept of robust ODE and prove that a robust ODE can be solved in polynomial time  (Theorem \ref{main-direction-one}). This is obtained, using an original method for solving ODE, optimising space, inspired by Savitch's theorem.  This provides one direction of all the above theorems.  The other direction is the object of the following sections, starting from Section  \ref{sectionPrevious}.  We first recall some previous results on discrete ODEs in  Section \ref{sectionPrevious}.
Using extensions of constructions from \cite{BlancBournezMFCS2023vantardise}, we then prove that we can simulate a Turing machine using robust continuous ODEs in Section \ref{sec:cstm}. This is obtained by simulating some discrete ODEs using continuous ODEs, dealing with error corrections, and using the fact that the functions are robust to a controlled error. The main result of Section \ref{sectionPrevious} is Theorem \ref{th:trendeux:new}. It states we can simulate Turing machines robustly with continuous ODEs when space remains polynomial. 
This theorem leads to the proof of Theorem \ref{th:mainone} in Section  \ref{sec:cstm}. In Section \ref{sectionMain}, we prove Theorem \ref{th:main:twopSpace}. In Section \ref{sec:conclusion}, we conclude and discuss future works.

\paragraph*{Some basic concepts} \label{sectionBasics}


When we say that a function over the real is computable this is always in the sense of computable analysis: see e.g. \cite{Wei00,Ko91,brattka2008tutorial}. A reference book for issues related to complexity theory in computable analysis is \cite{Ko91}.  See also appendix \ref{appendix:ca}.

%
%
%
%

\section{Dynamical systems and associated complexity issues}\label{sectionDynSys}

\subsection{Dynamical systems}

Dynamical Systems are often used to model natural phenomena and in many applied fields. The last decades have seen an impressive use of computers in studying and analysing dynamical systems, with several visible theoretical breakthroughs. A famous, notable example is the discovery of strange attractors in models such as Lorenz attractors through numerical simulations \cite{Lorenz63}, with only 40 years later the mathematical proof of their existence \cite{tucker2002rigorous}.  The mathematical proof was obtained by checking some quantitative invariant holds through computer-certified computations. 

From a mathematical point of view, a  \emph{discrete-time dynamical system} is given by a set $D$, called \emph{domain} and some (possibly partial) function  $\tu u$ from $D$ to $D$. A trajectory the system 
 is a sequence $\tu f(t)$ evolving according to $\tu u$: that is $\tu f(t+1)= \tu u\left(\tu f(t)\right)$ for all $t$. 
A dynamical system can equivalently be described by its flow $\Phi$: by definition, $\Phi(\tu f_{0}, t)$ gives the position of the dynamics at time $t$, for an initial position $\tu f(0)=\tu f_{0}$. It satisfies the flow property 
\begin{equation}\label{flowpropertyd}
\Phi(\tu f_{0}, 0) = \tu f_{0} \quad \Phi(\tu f_0, t+t') = \Phi(\Phi(\tu f_0, t), t')
\end{equation}
for all $t,t'$.  The transition dynamic $\tu u$ can be recovered from the flow function since $\tu u(.)= \Phi(.,1)$. Hence, describing a dynamical system by its dynamic $\tu u$ or by a flow function is equivalent. 
All of this can be parametrised by some parameter $\tu x$: $\tu u$ is also some function of $\tu x$, and $\tu f(t+1)= \tu u\left(\tu f(t),\tu x\right)$ for all $t$, and the flow function is $\Phi_{\tu x}(\tu f_{0},t)$. 
	
Up to that point, we were considering discrete-time systems, but we can also consider continuous-time dynamical systems: a \emph{continuous-time dynamical system}  is given by a
set $D \subseteq \mathbb{R}^d$ and some ODE  of the form ($\tu f'$ denotes the derivative with respect to time variable $t$) 
\begin{equation} \label{mycontoded}
 \tu f'= \tu u( {\tu f}(t))
\end{equation}
on $D$. A trajectory starting from $\tu f_{0}$ is a solution of the associated Initial Value Problem (IVP), with $\tu f(0)=\tu f_{0}$.  If we consider  $\Phi(\tu f_0, t)$ as giving the value of the solution at time $t$, starting from $\tu f_{0}$, it still satisfies the flow property \eqref{flowpropertyd}. Assuming sufficient regularity on $\Phi$ (namely that it is differentiable), $\Phi$ satisfies some ODE, and hence giving a dynamical system is equivalent to giving its flow: $\tu u(.)$ can be recovered by $\tu u'=\Phi'(.,0)$. Here, we can still consider that all of this is parameterised by some $\tu x$. 

Hence, in a very general view, a dynamical system is given by some function $\Phi$ satisfying the flow property. The fact that it is continuous time or discrete time, is related to the nature of its time (i.e. second)  variable, that belongs to $\N$ or $\Z$ in the latter case, and to $\R$ in the former. See \cite{HSD03} for a monography on the theory of dynamical systems from a mathematical point of view.  
 
 In all previous discussions, we considered homogeneous dynamics in the sense that $\tu u$ was only a function of $\tu f(t)$ and not also of $t$: but, for example, for discrete-time, to cover the case $\tu f(t+1)= \tu u\left(t,f(t)\right)$ for all $t$, it is sufficient to consider $\overline{\tu f}(t)=(\tu f(t),t)$, and we come back to the previous settings, as we can write $\overline{\tu f}(t+1)= \overline{\tu u}(\overline{\tu f}(t))$, with $\overline{\tu u}(\tu f,t)=(\tu u(\tu f,t),t+1)$, that is a homogeneous dynamic.

As this is well-known, dynamical systems exhibit a very rich class of possible behaviours. In particular, many dynamical systems are chaotic. We do not intend here to recall how chaoticity can be defined in mathematical terms (see \cite{Dev89a,HSD03}), but this includes at least high sensitivity to initial conditions. From our point of view, we just need to say this leads, in practice, to high unpredictability in the long run.  As observed in \cite{rojas2023algorithmic}, while countless papers exist in the literature about computations in dynamical systems, only a small fraction of them address the problem rigorously; i.e., how far is the sought actual quantity from the computed one? And can a computer perform such computation up to a very small pre-specified error? 
An important example where these questions are of interest is given by the reachability problem: given a (finite)  description of a dynamical  system, a description of its initial state, and the description of some ``unsafe'' states, the question is to tell whether a trajectory can reach an unsafe state.  

In the long run, dynamical systems may exhibit attractors.  Although there is a clear agreement about this intuitive concept, corresponding to the set of points to which most points evolve, it is a hard task to provide a mathematical definition covering all cases, and we refer to \cite{milnor1985concept} for discussions of many possible ways of defining this concept, and their relations. In the general case, the attractors of dynamical systems, even if the system is very simple, can be very rich. We refer to \cite{rojas2023algorithmic} for a characterisation of the hardness of computing attractors from a computable analysis point of view. We will somehow restrict to a very robust one (numerically stable ones) for which the problem is tractable. Somehow, our results state that the uncomputability discussed in \cite{rojas2023algorithmic} is intrinsically due to the non-numerical stability of the considered dynamical systems there.

\begin{example}
The van der Pol equation gives a very famous example of a dynamical system admitting an attractor 
$
y''-\mu\left(1-y^2\right) x' +y=0
$
where $\mu$ is a parameter. 
Introducing $z=y-y^3 / 3-\dot{y} / \mu$, it can be also described by $y'=\mu\left(y-\frac{1}{3} y^3-z\right)$, $z'=\frac{1}{\mu} y$. 
%
It is a classical mathematical exercise to prove that the system has a limit cycle: this is usually done as an application of Liénard's theorem, which is established by studying the flow of the dynamics, using qualitative arguments based on some formal mathematical statements: see, e.g. \cite{HSD03}. 
\end{example}

%
Formally, a  \emph{discrete-time dynamical system} is given by a set $D$, called \emph{domain} and some (possibly partial) function  $\tu u$ from $D$ to $D$. A trajectory the system 
 is a sequence $\tu f(t)$ evolving according to $\tu u$: that is $\tu f(t+1)= \tu u\left(\tu f(t)\right)$ for all $t$. 
 A \emph{continuous-time dynamical system}  is given by a
set $D \subseteq \mathbb{R}^d$ and some ODE of the form 
\begin{equation} \label{mycontode}
 \tu f'= \tu u( {\tu f}(t))
\end{equation}
on $D$. A trajectory starting from $\tu f_{0}$ is a solution of the associated Initial Value Problem (IVP), with $\tu f(0)=\tu f_{0}$. 
A dynamical system can equivalently be described by its flow: $\Phi(\tu f_{0}, t)$ gives the position of the dynamics at time $t$, for an initial position $\tu f(0)=\tu f_{0}$. It satisfies the flow property 
\begin{equation}\label{flowproperty}
\Phi(\tu f_{0}, 0) = \tu f_{0} \quad \Phi(\tu f_0, t+t') = \Phi(\Phi(\tu f_0, t), t').
\end{equation}
The dynamics or the flow can be parametrised by some parameter $\tu x$: $\tu u$ is also some function of $\tu x$  and the flow function is $\Phi_{\tu x}(\tu f_{0},t)$.

%
%
%
%
%
%
%
%
%
%
%

%

In the long run, dynamical systems may exhibit attractors.  We refer to \cite{milnor1985concept} for discussions of many possible ways of defining this concept, and  to \cite{rojas2023algorithmic} for a characterisation of the hardness of computing attractors from a computable analysis point of view. Somehow, our coming results state that the uncomputability discussed in \cite{rojas2023algorithmic} is intrinsically due to the non-numerical stability of the considered dynamical systems there.

\subsection{Some complexity results  on graphs}

We need to discuss the hardness of solving IVP, or equivalently of computing $\Phi(y, t)$. 
For pedagogical reasons, we first discuss the case of a simple setting, namely the case of a (deterministic) directed graph. Indeed, observe that a discrete-time dynamical system $(D,u)$ can also be seen as a particular (deterministic) directed graph $G=(V,\rightarrow)$, where, in the general case, $V$ is not necessarily finite: $G$ corresponds to $V=D$ and $\rightarrow$ to the graph of the function $u$, i.e. $\tu x_{t} \to \tu x_{t+1}$ iff ${\tu x}_{t+1}= \tu u\left({\tu x}_t\right)$. The obtained graph is deterministic because any vertex has an outdegree $1$. Starting from some point $\tu x_{0}$, there is at most one possible path, and consequently, for a given time $T$, we can talk about its position at time $t$, i.e. $\Phi(\tu x_{0},T)$ is $T$th element of this path: (as usual in complexity theory, the length of some integer $x$ is the length of its binary representation, denoted by $\ell(x)$). 

\SHORTERICALP{
	We then come to the space complexity of computing path, i.e. of computing $\Phi(\tu x_{0},T)$, when $V$ is finite (or equivalently when $V=D$ can be abstracted by some finite set). As usual in complexity theory, the length of some integer is the length of its binary representation. 
}

\begin{proposition}[The case of finite graphs] \label{corosuccint}
	Let $s(n) \ge \log(n)$ be space-constructible.
	Assume the vertices of $G=(V,\rightarrow)$ can be encoded in binary using words of length $s(n)$. Assume the relation $\rightarrow$ is decidable using a space polynomial in $s(n)$.
	Then,
	\begin{itemize}
		\item  given the encoding of $\tu u \in V$ and of $\tu v \in V$, we can decide whether there is some path from $\tu u$ to $\tu v$,  in a space polynomial in $s(n)$.
		\item given the encoding of $\tu u \in V$, and integer $T$ in binary, we can compute $\Phi(\tu u,T)$, in a space polynomial in $s(n)$ and the length of $T$.
	\end{itemize}
\end{proposition}

The second item is even a characterisation of the complexity of the problem. Indeed, the converse is true: If,  given the encoding of $\tu u \in V$, and integer $T$ in binary, we can compute $\Phi(\tu u,T)$, in a space polynomial in $s(n)$ and the length of $T$, then as $\rightarrow$ is given by $\Phi(.,1)$, then $\rightarrow$ is decidable using a space polynomial in $s(n)$.

\begin{proof}
	It is well-known that for finite graphs, 	given a directed graph $G=(V,\rightarrow)$ and some vertices $\tu u,\tu v \in V$, determine whether there is some path between $\tu u$ and $\tu v$ in $G$, denoted by $\tu u \stackrel{*}{\rightarrow} \tu v$ is in $\NLOGSPACE$: the rough idea is to guess non-deterministically the intermediate nodes. The formal proof is detailed in \cite{Sip97}.  The same algorithm, working over representations of vertices, when vertices are encoded using words of length $s(n)$ will work in $\NSPACE(s(n))$ (with the addition of the binary encoding of $T$ if for the second item, if it bigger than $s(n)$). 
	We then observe that $\NSPACE(s(n))=\sSPACE(s(n))$ from Savitch's theorem, recalled below. 
\end{proof}

\begin{theorem}[{Savitch's theorem, \cite[Theorem 8.5]{Sip97}}] \label{savitch} 
	For any space-constructible\footnote{As proved in \cite{Sip97}, this hypothesis can be avoided, at the price of a slightly more complicated proof.} function $s: \N \to \N$ with $s(n) \ge \log n$, we have 
	$\NSPACE(s(n)) \subseteq \sSPACE(s^{2}(n))$.
\end{theorem}

Recall that the key argument of the proof of Theorem \ref{savitch} is to express the question as a recursive procedure (expressing reachability in less than $2^{t}$ steps, called ${CANYIELD(\tu x}, \tu y ,t)$ in \cite{Sip97}) guaranteeing the required space complexity: we write that relation $CANYIELD(\tu x,\tu y, t)$  is relation $\tu x \rightarrow \tu y$ when $t=1$, and is relation $\exists \tu z$ such that $CANYIELD(\tu x,\tu z, t/2)$ and $CANYIELD(\tu z,\tu y, t/2)$ otherwise. If one prefers,  this can also be understood as ``guessing'' some intermediate node $\tu z$. 

\begin{remark}[Attractor point of view] We presented the above statement in terms of computing the flow $\Phi(\tu x,T)$. This could alternatively be interpreted in terms of attractors. Indeed, when the above hypothesis holds, then dynamics is captured by a graph. In the long run, in particular if $T$ is greater than the number of vertices, any trajectory loops (i.e. reaches an attractor).  
	The above statement could then also be read as the fact that such an attractor is then polynomial space computable.
\end{remark}


\subsection{Solving  efficiently ODEs: what is known}

This idea leads to an original method for solving ODEs. At least, this is original for the numerical analysis literature, as far as we know.

We  review what is known about the complexity of ODE solving. A more complete survey is \cite{GracaZhongHandbook}. First, it is important to distinguish the case where we want to solve the ODE on a bounded  (hence a compact) domain, from the case of the full domain $\R$: in the latter case,  we might ask questions about the evolution of the system on the long run, which is harder. 
Over a compact domain, it is known that there exists some polynomial-time computable function 
$u:   [-1,1] \times [0,1] \rightarrow \R$ such that $f'=u(f,t)$ has no computable solution, even over $[0,\delta]$, for any $\delta>0$: see \cite{Ko83}, extending \cite{PouRic79,Abe71}. The involved ODE has no unique solution. It is known over compact or non-compact domains that if unicity holds, then its solution is computable \cite{collins2008effectivesimpl,collins2009effective,Ruo96}.  However, the complexity can be arbitrarily high \cite{Ko91,Miller70}. 

If we want to get to tractability, then some regularity hypotheses must be assumed. A classical hypothesis is to assume the ODE to be Lipschitz.  

Over a compact domain, it has been observed in several references (see e.g. \cite{Ko91}) that a careful analysis of Euler's method proves that, if  $u:B(0,1)  \times [0,1]  \rightarrow \mathbb{R}^n$, with $B(0,1) \subseteq \mathbb{R}^n$, is a polynomial time computable (right-)Lipschitz function then any solution $f:[0,1] \rightarrow B(0,1)$ of   $f'=u(f,t)$  must be polynomial-space computable: see the discussions around Theorem 3.2 in \cite{GracaZhongHandbook} with the several references.   Kawamura has proved in \cite{kawamura2009lipschitz} that there exists a polynomial-time computable function $u: [-1,1] \times [0,1]  \rightarrow \mathbb{R}$, which satisfies a Lipschitz condition, such that the unique solution $f: [0,1] \rightarrow \mathbb{R}$ takes values in $[-1,1]$ and computing it leads to a  $\Pspace
$-complete problem. Hence, the question of solving ODEs over a compact domain in polynomial time is exactly the question $\Ptime=\Pspace$ \cite{kawamura2009lipschitz},  even for $\mathcal{C}^{\infty}$-functions \cite{kawamura2014computational}.

However, all these results are over compact domains, and dealing with non-compact domains, i.e. in the long run, is harder.  $\Pspace$ membership is not true, without stronger hypotheses. The difficulty comes from the possibility of simulating any Turing machine by some finite-dimensional polynomial ODE \cite{dsg06a} over a non-compact domain.  This leads to many undecidability results for analytic, and even very simple ODEs. For example, it follows that there is an analytic and computable function $u: \mathbb{R} \rightarrow \mathbb{R}$ such that the unique solution of the associated homogeneous ODE is defined on a non-computable maximal interval of existence \cite{dsg06a}.   Futhermore, if we consider $f_1(t)=e^t$, and  $f_{i+1}(t)=e^{y_i(t)-1}$ then $f_d(t)$ is $e^{e^{. e^{e^t}-1}}-1$, while all these functions are solutions of a simple polynomial ODE over $\R^{d}$, namely $f'_{1}(t)=f_{1}(t)$ and $f'_{d}(t)= f_{1} \dots f_{d}(t)$: i.e. a solution can grow faster than a tower of exponentials in the description of the ODE, and hence is necessarily intractable for time or space: see the discussion in \cite[Section 3.2]{GracaZhongHandbook}. A possible way to analyse efficiency is then to analyse the complexity of the solution assuming a bound on the growth of the function (i.e. using parameterised complexity). It was proved in \cite{bournez2012complexity} that one can solve a polynomial ODE in polynomial time assuming a bound on $\tu Y(T)=\max _{0 \leq t \leq T}\|\tu f(t)\|$. 
This result can be extended to non-polynomial ODEs assuming polynomial-time computability of the higher derivatives of $\tu f$ and an appropriate (polynomial) bound on the growth of those derivatives \cite{bournez2012complexity}. 

The result for polynomial ODEs was later improved in \cite{PoulyGraca16}, where it is proved that the time $T$ and parameter $\tu Y$ can be replaced by a single parameter, namely the length of the curve for polynomial ODEs.  Furthermore, this parameter does not need to be given as input to the algorithm.  This is a key argument for one direction of the moto ``time complexity = length'' we mentioned several times. 
This is obtained using a non-classical method, from the point of view of classical numerical analysis: this is not a fixed-order numerical method, but somehow a method whose order is chosen as a function of the inputs. A similar method was independently proposed in \cite{HolgerThiesPhD,kawamura2018parameterized}, considering parametrised complexity for analytic functions.

To get polynomial-time complexity over a non-compact domain, it is mandatory not to use most classical methods from numerical analysis. Indeed, from the general theory of numerical methods, presented for example in \cite{Dem96}, every such numerical method come with a given fixed order $k$, and from the general theory of step-based methods, interval $[0,T]$ is divided into $N$ steps, each of width $h_{n} \le h$, and the error $\theta_{n}$ at step $n$ can be bounded by $\theta_{n+1} \le (1+\Lambda h_n) \theta_n + |\epsilon_n|$: there is a multiplicative error due to the Lipschitz constant $\Lambda$ of $u$, and some numerical additive error $\epsilon_{n}$ due to the used precision. 
Using Discrete Grönwall Lemma (see \cite[page 213]{Dem96}, or lemma \ref{discretegronwall} in appendix)  the final error satisfies between computed approximation $\widetilde{f}_{n}$ and exact solution ${f}_{n}$ at step $n$ satisfies $
\max _{0 \leq n \leq N}\left|\widetilde{f}_{n}-f_{n}\right| \leq 
e^{\Lambda T} \left|\widetilde{f}_{0}-f_{0}\right| + \frac{e^{\Lambda T}-1}{\Lambda}  \left| \max \frac{\varepsilon_{n}}{h} \right|$. If $N=T/h$, we have 
$\max _{0 \leq n \leq N}\left|\widetilde{f}_{n}-f_{n}\right| \leq 
e^{\Lambda T} \left|\widetilde{f}_{0}-f_{0}\right| + \frac{e^{\Lambda T}-1}{\Lambda}  \left| \max \frac{\varepsilon_{n}}{h} \right|$. To go to zero, the last factor is chosen to be of the form $\exp(-\mathcal{O}(T))$.

The same problem happens when discussing space complexity: a non-classical method is required to guarantee polynomial space complexity in the long run (i.e. on the non-compact domain). As far as we know, no such method has yet been proposed, and this is the purpose of the coming subsection. Actually, for space complexity,  in addition to all the problems mentioned, in all the above space or time analyses, the problem is that the complexity is (possibly implicitly) dependent on the Lipschitz constant or the length of the solution. In a system as simple as linear dynamics, the state at time $T$ depends in Lipchitz way from the state at time $0$, and the number or additional bits required to guarantee some precision $2^{-n}$ growth linearly with $T$. But the problem is that in a space polynomial in the input size, $T$ has no reason to remain polynomial (consider, for example, a system simulating a Turing machine, as we will consider soon). Hence, the required precision is possibly exponential in the input size.

The above comment can be interpreted informally as the fact that ``most'' (this could be ``generic'' in the sense of \cite{rojas2023algorithmic}, i.e. (effective) descriptive theory)
dynamical systems are intrinsically unstable, and an error method introduced at some step can make the method unavoidably incorrect in the long run unless we have a means to ``guess'' what will happen. 

 \begin{remark}[Attractor point of view] We presented the above statement in terms of computing the flow $\Phi(\tu f_{0},t)$. But, this could alternatively be interpreted in terms of attractors. The point is that computing the attractors of a given dynamical system is hard in general, as this involves long-run behaviours. This explains all the undecidability results obtained in \cite{rojas2023algorithmic}, even for very simple dynamics. However, as we will see, this is also explained by the fact that the latter paper is discussing numerically unstable systems. 
\end{remark}

\subsection{Solving  efficiently ODEs: a  space efficient method} 
 
This leads to an alternative approach to optimize space complexity: this can be seen as either using a non-deterministic algorithm that ``guesses'' the correct intermediate positions of the dynamics or, from the proof of Savitch's theorem approach, as an original recursive method to solve ODEs. As far as we know, we have never seen such a method discussed in the literature for solving ODEs. 

Concretely: from the flow property, a strategy to compute $\Phi(\tu f_{0},T)$ is either to use a particular numerical method if $T$ is small, says smaller than $\Delta>0$. Otherwise, we know that $\Phi(\tu f_{0},T)=\Phi(\tu z,T/2)$, where $\tu z= \Phi(\tu f_{0},T/2)$. This always holds, so if we can compute both quantities, we will solve the problem. The difficulty is that we cannot precisely compute $\tu z$ in practice, but some numerical approximation $\widetilde{\tu z}$. If the system is numerically stable, we may assume this strategy works. The case when this strategy will not work is if the trajectory starting from $\widetilde{\tu z}$, for the second half of the work from time $T/2$ to $T$,  has a behaviour different from the one starting in $\tu z$: in other words, if there is a high instability somewhere, namely in $\tu z$. 

This leads to the following concept: we write $\tu a=_{n} \tu b$ for $\| \tu a-\tu b \| \le 2^{-n}$ for conciseness.

\begin{definition}[Robust (continuous) ODE]\label{defRobODE}
	
	A function $\tu f : \R \rightarrow \R$ is robustly ODE definable (from initial condition $\tu g$, and dynamic $\tu u$) if
	\begin{enumerate}
		\item it corresponds to the
		solution of the following continuous ODE:
		\begin{equation}\label{contODE}
			\tu f(0,\tu x) 
			= \tu g(\tu x)   \quad  and \quad
			\frac{\partial \tu f(t,\tu x)}{\partial t} 
			=   \tu u(\tu f(t,\tu x), \tu h(t,\tu x),
			t,\tu x),
		\end{equation}

		\item and there is some rational $\Delta >0$, 	and some polynomial $p$ such that  the schema \ref{contODE} is (polynomially) numerically stable on $[0, \Delta]$:
		 for all integer $n$, considering $\epsilon(n)=p(n+\ell(\tu x))$ we can compute $\tu f(t,\tu x)$ by working a precision $\epsilon(n)$:
			if you consider any solution of  
{
$\tilde{\tu x} =_{\epsilon(n)} \tu x \quad and \quad \tilde{\tu h}(t,\tilde{\tu x}) =_{\epsilon(n)} {\tu h}(t,\tilde{\tu x})$,
} and 
$
 \tilde{\tu f}(0,\tilde{\tu x}) =_{\epsilon(n)}
\tu g(\tu x)$  and 
$
 \frac{\partial \tilde{\tu f}(t,\tilde{\tu x})}{\partial t} =_{\epsilon(n)}
   \tu u(\tilde{\tu f}(t,\tilde{\tu x}), \tilde{\tu h}(t,\tilde{\tu x}),
t,\tilde{\tu x}) 
$
then
$\tilde{\tu f}(t,\tilde{\tu x}) =_{\epsilon(n)} \tu f(t,\tu x)$ when  $0 \le t \le \Delta$.

	\item For $t \ge \Delta$,  we can compute $\tu f(t,\tu x)$ by computing some approximation $\widetilde{\tu f(t/2,\tu x)}$ of 
 $\tu f(t/2,\tu x)$ at precision  $\epsilon(n)$, i.e. of $\Phi(\tu g(x),t/2)$, and then
 some approximation of  $\Phi(\widetilde{\tu f(t/2,\tu y)},t/2)$, working at precision  $\epsilon(n)$.	 
		

	\end{enumerate}	
\end{definition}

\begin{theorem} \label{main-direction-one}
	Consider an IVP as in the previous definition. If $\tu g, ~\tu h ~\text{and} ~\tu u$ are computable in polynomial space, then the solution $\tu f$ can be computed in polynomial space. 
\end{theorem}

\begin{proof}
	From definitions and above arguments,  all bits of $\Phi(\tu y,t)$ can be computed non-deterministically with precision $n$ (i.e. at $2^{-n}$) using computations with precision $\epsilon(n)$, hence is in  $\NPspace=\Pspace$. From the argument of the proof of Savitch's theorem, this can also be turned into a deterministic polynomial space recursive algorithm. 
\end{proof}

The above theorem is the key argument to obtain one direction of our main theorems.  We now go in the reverse direction. This requires talking about discrete ODEs, and some previous constructions. 

\section{Discrete ODEs: some previous results and constructions}
\label{sectionPrevious}

\subsection{Preliminary}

%
%

We will use the concept of discrete ODE defined as follows (notice that we will write $\frac{\delta \tu f}{\delta n}$ for discrete derivation, by opposition of the classical $\frac{\partial \tu f}{\partial n}$ to help to distinguish discrete vs continuous ODEs. \SHORTERICALP{: i.e. we use symbol $\delta$ for discrete, by opposition of symbol $\partial$. In some contexts, when this is clear, we will also write $\tu f'$ for $\frac{\delta \tu f}{\delta n}$ or $\frac{\partial \tu f}{\partial n}$).})

\begin{definition}[Discrete derivation, notation $\delta$]
	For $\tu f: \N \to \R^{d}  \rightarrow \R^{d'} $, the \emph{discrete derivation} of $\tu f$ is 
	$\frac{\delta \tu f}{\delta n}(n,\tu x) = \tu f(n+1,\tu x) - \tu f(n,\tu x)$.
\end{definition}


\subsection{Algebraic characterisation with discrete ODEs: state of the art}
 \label{sectionPrevioussoa}
 
In this subsection, we review some of the results already obtained using discrete ODEs. 

\begin{remark}
Notice that we do not need any of these statements directly, even if we will sometimes reuse some of their constructions (and some of their ideas). 
\SHORTERICALP{Consequently, a reader can skip this section. However, of course, it is important to state what was known before the coming constructions. }
\end{remark}

\newcommand\myparagraph[1]{\noindent {\textbf{#1}}:\\ }

\myparagraph{Characterising $\Ptime$ over the integers}
The concept of derivation along the length was introduced in \cite{MFCS2019}. 
\SHORTERICALP{
First, we must define the notion of derivation along the length. We use the notation $\dderivl{\tu f(x,\tu y)}$ which corresponds to the derivation of $\tu f$ along the length function: given some function $\lengt:\N^{p+1} \rightarrow \Z$ and in particular for the case where $\lengt(t,\tu y)=\ell(t)$,
\begin{align}\label{lode}
	\dderivL{\tu f(t,\tu x)}= \dderiv{\tu f(t,\tu x)}{\lengt(t,\tu
		x)} = \tu h(\tu f(t,\tu x),t,\tu x)
\end{align}
is a formal synonym for
$ \tu f(t+1,\tu x)= \tu f(t,\tu x) + (\lengt(t+1,\tu x)-\lengt(t,\tu x)) \cdot
\tu h(\tu f(t,\tu x),t,\tu x).$
}
A characterisation of $\FPtime$ for functions over the integers has then been obtained in \cite{MFCS2019}:
\begin{theorem}[Functions over the integers \cite{MFCS2019}] 
	$\linearderivlength \cap \N^{\N}= \FPtime \cap \N^{\N}$, for $\linearderivlength =$ $ [\mathbf{0},\mathbf{1},\projection{i}{k}, \length{x}, \plus, \minus, \times, \sign{x} \ ; \composition, \linearLODE],$
	with $\projection{i}{k}$ the projection function, 
	and $\sign{x}$ is 0 for $x<0$ and 1 for $x>0$.
\end{theorem}

\myparagraph{Toward the real numbers: characterising real sequences}
Later, the authors of \cite{BlancBournezMCU22vantardise} introduced
\begin{definition}[Operation $\MANONlim$] Given $\tilde{\tu f}:\R^{d}\times \N \to \R^{d'} \in \manonclass$ such that
	for all $\tu x \in \R^{d}$, $n \in \N$,
	$\|\tilde{\tu f}(\tu x,2^{n}) - \tu f(\tu x) \| \le 2^{-n}$ for some function $\tu f$, then 
	$\MANONlim(\tilde{\tu f})$ is the (uniquely defined) corresponding function  $\tu f: \R^{d} \to \R^{d'}$.
\end{definition}
and then considered the class  $$\manonclasslim = [\mathbf{0},\mathbf{1},\projection{i}{k},   \length{x}, \plus, \minus, \times,\signb{x},\frac{x}{2};{\composition, \linearLODE, \MANONlim}],$$
with $\signb{x}$ a sigmoid valuing 0 when $x<\frac{1}{4}$ and 1 when $x>\frac{3}{4}$. 
They proved this provides a characterisation of functions from $\N$ to $\R$ computable in polynomial time. 

\begin{theorem}[Sequences of reals \cite{BlancBournezMCU22vantardise}]
	$\manonclasslim = \FPtime \cap \R^\N$.
\end{theorem}

\myparagraph{Characterisation of $\Ptime$ and $\Pspace$ for functions over the real with discrete ODEs}
The same authors later succeeded in obtaining a characterisation of functions over the real computable in polynomial time and even space. 

%
\begin{theorem}[$\FPtime$, , Generic functions over the reals \cite{BlancBournezMFCS2023vantardise}] \label{th:main:twop} 
	$\manonclasslighttanhlim \cap \R^{\N^{d} \times \R^{d'}} = \FPtime \cap \R^{\N^{d} \times \R^{d}} $, with 
	$$\manonclasslighttanhlim = [\mathbf{0},\mathbf{1},\projection{i}{k},   \length{x}, \plus, \minus,\tanh,\frac{x}{2},\frac{x}{3};{\composition, \linearLODE, \MANONlim}].$$
\end{theorem}
%
%
Consider the following schema:
\begin{definition}[Robust Discrete ODE \cite{BlancBournezMFCS2023vantardise}]\label{def:roblinear lengt ODE} \label{schema:space}
	A bounded function $\tu f$ is robustly ODE definable if:
	\begin{enumerate}
		\item it corresponds to the
		solution of the following discrete ODE:
		\begin{equation} \label{dynamique}
			\tu f(0,\tu x) 
			= \tu g(\tu x)   \quad  and \quad
			\frac{\delta \tu f(t,\tu x)}{\delta  t} 
			=   \tu u(\tu f(t,\tu x), \tu h(t,\tu x),
			t,\tu x),
		\end{equation}
		\item where  the schema \eqref{dynamique} is (polynomially) numerically stable: there exists some polynomial $p$ such that, for all integer $n$, writing $\epsilon(n)=p(n+\ell(\tu y))$, if you consider any solution of  
{
$\tilde{\tu y} =_{\epsilon(n)} \tu y \quad and \quad \tilde{\tu h}(x,\tilde{\tu y}) =_{\epsilon(n)} {\tu h}(x,\tilde{\tu y})$,
} and 
$
 \tilde{\tu f}(0,\tilde{\tu y}) =_{\epsilon(n)}
\tu g(\tu y)$  and 
$
 \frac{\partial \tilde{\tu f}(x,\tilde{\tu y})}{\partial x} =_{\epsilon(n)}
   \tu u(\tilde{\tu f}(x,\tilde{\tu y}), \tilde{\tu h}(x,\tilde{\tu y}),
x,\tilde{\tu y}) 
$
then
$\tilde{\tu f}(x,\tilde{\tu y}) =_{\epsilon(n)} \tu f(x,\tu y)$. 

	\end{enumerate}
\end{definition}

A robust discrete ODE is said to be  \textit{linear} if $\tu u$ is essentially linear in $\tu f$ and $\tu h$.

Consider 
$$\spaceclasstanhlim = [\mathbf{0},\mathbf{1},\projection{i}{k},   \length{x}, \plus, \minus,\tanh,\frac{x}{2},\frac{x}{3};{composition, robust~linear~ODE, \MANONlim}].$$

\begin{theorem}[$\FPspace$, Generic functions over the reals \cite{BlancBournezMFCS2023vantardise}] 
	{$\spaceclasstanhlim \cap \R^{\R}  =  \FPspace \cap \R^{\R} $} \\
	More generally: \quad \quad \quad  \quad  \hspace{0.1cm}
	{$\spaceclasstanhlim \cap \R^{\N^{d} \times \R^{d'}}  = \FPspace \cap \R^{\N^{d} \times \R^{d}} $.} 
\end{theorem}

\SHORTERICALP{This ends our state of the art. }
Notice that previous classes mix functions with integer and real arguments. Furthermore, they all involve some various types of discrete ODEs. We need to avoid all these issues, as we consider only continuous ODEs.

\subsection{Simulating a discrete ODE using a continuous ODE}
\label{ideal:branicky} 

We first prove that it is possible to simulate a discrete ODE with a continuous ODE. The underlying idea can be attributed to \cite{Bra95}, and has been improved in many ways by several authors. We present here the basic ideas, reformulated in our context. A more precise analysis  will come (Proposition \ref{the:manon:trick}).

\begin{definition}["Ideal iteration trick", \cite{Bra95}] \label{defBranicky}
	Consider the following initial value problem for a discrete ODE, given by functions $\tu g$ and $\tu u$:
	\begin{equation} \label{dode}
		\left\{
		\begin{aligned}
			\tu f(0,\tu x) &= \tu g(\tu x)\\
			\frac{\delta \tu f}{\delta t}(t,\tu x) &= \tu u (\tu f(t,\tu x),t,\tu x) \\
		\end{aligned}
		\right.
	\end{equation}	
	
	Then, let $\tu G(\tu v, t, \tu x) = \tu u(\tu v,t, \tu x) + \tu v$, 
	and consider the (continuous) IVP: 
	\begin{equation} \label{code}
		\left\{
		\begin{aligned}
			&\tu y_1(0,\tu x) = \tu y_2(0,\tu x) = \tu g(\tu x)\\
			&\tu y_1' = c (\tu G(r(\tu y_2),r(t),\tu x) - \tu y_1 )^3 \theta(\sin(2\pi t)) \\
			&\tu y_2' = c (r(\tu y_1) - \tu y_2 )^3 \theta(- \sin(2\pi t))\\
		\end{aligned}
		\right.
	\end{equation}
	where $c$ a constant, $\theta(x)=0$ if $x\leq 0$ and $\theta(x) >0$ if $x>0$.  We abusively write $r(\tu y)$ for the application of function  $r: \R \to \R$ componentwise on vector $\tu y$.  Here, $r$ is a rounding function: we mean, by construction, $G$ preserves the integers, and $r$ is a function that maps a real value close to some integer to this integer: assume, say,  that for $z \in [n-\frac14,n+\frac12]$, $r(z)=n$, for any integer $n \in \Z$.
\end{definition}

Then, the solution of continuous ODE \eqref{code} simulates in a continuous way the discrete ODE \eqref{dode}: Indeed,   $\tu y_1$ 
corresponds to the actual computation of the iterates  of $\tu G$ (and hence computes the successive values of $\tu f$)  and $\tu y_2$ acts as a ``memory" equation. Let us detail how it works.  
We denote by $t =_\epsilon z$ the fact that $\left| t-z \right| \leq \epsilon$.

\begin{remark}
	We describe here an ``ideal" computation, as $\theta(x)$ is exacly $0$ when
	$x\leq 0 $, and $r(z)$ is exactly some integer on suitable domains. 
	 Later in the paper, we will deal with a not-so-ideal $\theta$ and $r$. 
\end{remark}

Initially, $\tu f(0,\tu x) = \tu y_{1}(0,\tu x) = \tu y_{2}(0,\tu x) = \tu g(\tu x)$.
For $t \in [0,1/2]$, we have $\theta(- \sin(2\pi t))=0$, and hence 
$\tu y_2' = 0$, so $\tu y_2$ is fixed and kept at value $\tu g(\tu x)$ for $t \in[0,\frac{1}{2}]$. 
Consequently, for $t \in [0,1/2]$, $r(\tu y_{2})$ is also fixed and kept at value $\tu g(\tu x)$, and $r(t)$ is also fixed and kept at value $0$. Consequently, on this interval, 
if we write $C(t) = c\theta(\sin(2\pi t))$, then the dynamics of $\tu y_1$ is given by 
\begin{equation} \label{eq:de:branicky}
\tu y_1' = C (t) (\tu G(\tu g(\tu x),0,\tu x) - \tu y_1 )^3
\end{equation}

\begin{lemma}[Analysis of ODE \eqref{eq:de:branicky}] \label{bran:funda}
The solution $\tu y_{1}(t,\tu x)$ of ODE \eqref{eq:de:branicky} is converging to $G(\tu g(\tu x),0,\tu x)$ for any initial condition. Furthermore, for any initial condition $\tu y_{1}(0,\tu x) \neq G(\tu g(\tu x),0,\tu x)$, we have $\left\| \tu y_1(\frac{1}{2}, \tu x) - \tu G(\tu g(x),0,\tu x) \right\| \leq \frac{\sqrt{2}}{2\sqrt{\int_{0}^{\frac{1}{2}} C(z)dz}}$.
In particular, 
for any $m \in \N$, we can select constant $c$ such that for any initial condition $\tu y_{1}(0,\tu x)$, 
$$\left\| \tu y_1(\frac{1}{2}, \tu x) - \tu G(\tu g(x),0,\tu x) \right\| \le 2^{-m}.$$
\end{lemma}

\begin{proof}[Proof of Lemma \ref{bran:funda}, Adapted from \cite{Bra95}]
	If initially, or at any instant  $\tu y_{1}(0,\tu x)=G(\tu g(x),0,\tu x)$ then the result holds, as $\tu y'_{1}=0$, and $\tu y_{1}$ remains constant. Otherwise, we have 
	$$\frac{\tu y_{1}'}{\tu G(\tu g(\tu x),0,\tu x) - \tu y_1 )^3} = C (t).$$
	Integrating this equality between $0$ and $t$, we obtain  
	$$\frac{\tu 1}{2(\tu G(\tu g(\tu x),0,\tu x) - \tu y_1(t) )^2} -  \frac{\tu 1}{2(\tu G(\tu g(\tu x),0,\tu x) - \tu y_1(0) )^2}= \int_{0}^{t} C(z) dz,$$ hence $$\frac{\tu 1}{2(\tu G(\tu g(\tu x),0,\tu x) - \tu y_1(t) )^2}  \ge \int_{0}^{t} C(z) dz.$$ This yields the property.
\end{proof}

Consequently, $\tu y_1(t,\tu x)$ will approach $\tu G(\tu g(\tu x),0,\tu x)=\tu f(1,\tu x)$ on this interval.
 Thus, $\tu y_1(\frac{1}{2},\tu x) =_\epsilon \tu f(1,\tu x)$ and $\tu y_2(\frac{1}{2},\tu x) = \tu g(\tu x)$, for some $\epsilon > 0$, that we can consider less than $\frac14=2^{-2}$, by selecting a big enough constant $c$ (just taking $m=2$ above). 
At $t=\frac{1}{2}$, $\tu y_1$ will hence have simulated one step of discrete ODE \eqref{dode}. 

Now, for $t \in [\frac{1}{2}, 1]$ the roles of $\tu y_1$ and $\tu y_2$ are exchanged : $\tu y_1'(t, \tu x) = 0$, so $\tu y_1$ is kept fixed, $\tu y_2$ approaches $r(\tu y_{1})=\tu f(1, \tu x)$, thus 
 $\tu y_1(1,\tu x) =_\epsilon \tu y_2(1, x) =_\epsilon \tu f(1, \tu x)$.

By induction, from the same reasoning, we obtain that, for all $n\in\N$, $\tu y_1(n,\tu x) =_\epsilon \tu y_2(n,\tu x) =_\epsilon \tu f(n,\tu x)$, and actually, we also have $\tu y_1(t+\frac12,\tu x) =_\epsilon \tu y_2(t,\tu x) =_\epsilon \tu f(n,\tu x)$ for all $t \in [n,n+\frac12]$, for any integer $n$.



To implement such an ODE, we have to fix a function $\theta(x)$ with the above property. Taking 
$\relu(x)=\max(0, x)$ would satisfy it, but it is not a derivable function, and hence would not lead to a (classical) ODE. We could then take $\theta(x)=0$ for $x \le 0$, and $\exp(-1/x)$ for $x>0$. The point is that such a function is not real analytic. The base functions we consider in our class $\contClass$  are all real analytic, and real analytic functions are preserved by composition, so we cannot get such a function by compositions from our base functions. Futhermore, it is known that a real analytic function that is constant on some interval (we assumed it is $0$ for $x\le 0$!)  is constant. Hence the above-considered function $\theta(x)$ cannot be real analytic.  So, implementing this trick cannot be done directly using our base functions, using only compositions. 

In Proposition \ref{the:manon:trick}, we will do a similar construction, but dealing with errors and not exact function $\theta(z)$ and $r(x)$. 
Furthermore, here the purpose of function $r$ was to correct errors around integers, i.e. around $\N$: this will be possibly around other $\N\delta$ for some $\delta>0$.

\subsection{Encoding of Turing machines configurations}

Our proofs rely on some constructions from \cite{BlancBournezMFCS2023vantardise}. Concretely, we need to simulate the execution of a Turing machine (TM) $\mathcal{M}$ by some dynamical system over the reals. This requires to encode the configurations of a Turing machine into some real numbers. We recall some of the definitions and constructions from \cite{BlancBournezMFCS2023vantardise}.  

Consider a Turing machine defined by $\mathcal{M} = (\Sigma, Q, I, F, \delta)$, with $\Sigma$ the working alphabet, $Q$ the set of states, $I,F \subseteq Q$ respectively the sets of initial and final states, $\delta : Q \times \Sigma \rightarrow Q\times\Sigma\times\{\leftarrow, \rightarrow\}$ the transition function.
For some practical reasons, similar to the ones in \cite{BlancBournezMFCS2023vantardise}, we assume that the working alphabet is made of the symbols $1$ and $3$, and that the blank symbol is symbol $0$.

We explicit the encoding we will use. 
We assume $Q=\{0,1,\dots,|Q|-1\}$.  Let 
$$ \dots  l_{-k} l_{-k+1} \dots l_{-1} l_{0} r_0 r_1 \dots r_n .\dots$$
denote the content of the tape of the Turing machine $M$. In this representation, the head is in front of symbol $r_{0}$, and $l_i, r_{i} \in  \{0,1,3\}$ for all $i$.  Furthermore, we assume that there are no non-blank symbols between two blank symbols, i.e. that blank symbols, i.e. symbol $0$, can only be eventually on the right, or eventually on the left. 
Such a configuration $C$ can be denoted by $C=(q,l,r)$,
where $l,r \in \Sigma^{\omega}$ are words over alphabet
$\Sigma=\{0, 1,3\}$ and $q \in Q$ denotes the internal state of $M$.

Now, write: $\encodagemot: \Sigma^{\omega} \to \R$ for the function that
maps a word $w=w_{0} w_{1} w_{2} \dots$ to the dyadic  (hence real) number
$\encodagemot(w)= \sum_{n \geq 0} w_n \base^{-(n+1)}$.

The idea is that configuration $C$ can also be encoded by some element $\bar C=(q, \bar l,\bar r) \in \N \times \R^{2}$, by considering $\bar r= \encodagemot(r)$ and $\bar l=\encodagemot(l)$. 
In other words,  we encode the configuration of a bi-infinite tape Turing machine $M$ by real numbers using their radix \base{}  encoding, but using only digits $1$,$3$.
Notice that this lives in $Q \times [0,1]^{2}$. Denoting the image of $\encodagemot: \Sigma^{\omega} \to \R$ by $\Image$, this even lives in $Q \times \Image^{2}$. 

In other words, we consider the following encodings:
$\gamma_{config}(C) = (q, \bar{l}, \bar{r})$
with
$\bar{l} = l_0 4^{-1} + l_{-1}4^{-2} + \dots + l_{-k} 4^{-(k+1)} + \dots $ and
$ \bar{r} = r_0 4^{-1} + r_{1}4^{-2} + \dots + l_{n} 4^{-(n+1)} + \dots$.


\subsection{Revisiting some previous constructions}

We denote by $\contClasslight$ the algebra $[0, 1,\pi_i^k, +, -, \times, \tanh, \cos, \pi, \frac{x}{2}, \frac{x}{3}; \composition] $.
This is close to the class $$\manonclasslighttanh = [\mathbf{0},\mathbf{1},\projection{i}{k},   \length{x}, \plus, \minus,\tanh,\frac{x}{2},\frac{x}{3};{composition, linear~length~ODE}],$$ considered in \cite{BlancBournezMFCS2023Journal,BlancBournezMFCS2023vantardise}, but without the function $\length{x}$, and wihtout the possiblity of defining functions using linear~length~ODE (and with multiplication added). 

We will reuse some of the construction from \cite{BlancBournezMFCS2023vantardise} (some corrections and more details can be found in \cite{BlancBournezMFCS2023Journal}) but avoid systematically any use of linear length ODE and the length function $\ell(x)$. Furthermore, the class considered in \cite{BlancBournezMFCS2023vantardise} is mixing functions from the integers to the reals, and from the reals to the reals, and we need to keep only functions over the reals.
\SHORTERICALP{, so we need to avoid completely any discrete function that can be considered there, in particular any computation of $2^{m}$ from some $m$.}

The following was stated in  \cite[Lemma 19]{BlancBournezMFCS2023vantardise}.
\begin{lemma}\label{lemmeRelu}
	We denote by $Y(x,2^{m+2})$ the 
	function  $\frac{1+\tanh(2^{m+2} x)}{2} $. For all integer $m$, for all $x\in \R$,
	$| \relu(x) - xY(x, 2^{m+2})| \leq 2^{-m}$, where $\relu(x)=\max(0, x)$. 
\end{lemma}
First, we observe that considering $Y(x,z)=\frac{1+\tanh(4 x z)}{2}$ would yield a function in $\contClasslight$ with the same property: we avoid the computation of $2^{m}$ by a substitution of a variable, and using a multiplication. We then write $\relutanh(Y,x)$ for $xY(x,z)$: we have $| \relutanh(2^{m},x) -\relu(x)| \leq 2^{-m}$.

In particular, this was used to prove we can uniformly approximate the continuous sigmoid functions (when $1/(b-a)$ is in $\manonclasslighttanh$) defined as: $\sig(a,b,x) = 0$ whenever $w \leq a$, $\frac{x-a}{b-a}$ whenever $a \le x \le b$, and $1$ whenever $b \leq x$. The above trick provides a new version of  \cite[Lemma 20]{BlancBournezMFCS2023vantardise}. 

\begin{lemma}[Uniform approximation of any piecewise continuous sigmoid] \label{lem:solution:a:tout}
	Assume $a,b,\frac1{b-a}$ is in $\contClasslight$. Then there is some function 
	$\sigtanh(z,a,b,x) \in \contClasslight$  
	such that for all integer $m$, $|\sigtanh(2^{m},a,b,x) - \sig(a,b,x)|\le 2^{-m}$.
\end{lemma}

\begin{proof}
	Take $\sigtanh(z, a, b, x)=  \frac{(x-a) Y(x-a, z2^{1+c}) - (x-b) Y(x-b, z2^{1+c}) }{b-a}$. 
	observing that $(b-a) \sig(a,b,x) = {\relu(x-a)-\relu(x-b)}$. 
	From triangle inequality,  it will hold, choosing   $c$ with $\frac1{b-a} \le 2^{c}$.
\end{proof}

The authors of \cite{BlancBournezMFCS2023vantardise} proved the existence of some function corresponding to a continuous (controlled) approximation of the fractional part function:

\begin{theorem}[{\cite[Lemma 28]{BlancBournezMFCS2023vantardise}}]  \label{thXi}
	There exists some function $\xi: \N^{2}\to \R$ in $\manonclasslighttanh$ such that for all $n,m\in \N$ and $x\in [- 2^{n} , 2^{n}]$, whenever $ x \in [\lfloor x \rfloor + \frac{1}{8}, \lfloor x \rfloor + \frac{7}{8}] $, $$\left|\xi(2^m,2^{n},x)  - \{ x -\frac18 \}\right| \le 2^{-m}.$$
\end{theorem}

We say that some real function is a real extension of a function over the integers if they coincide for integer arguments 
It is not clear that we have a real extension of $\xi$ in our algebra $\contClasslight$, but if we add a real extension of such a function, from the proof of {\cite[Corollary 22]{BlancBournezMFCS2023vantardise}}, we obtain the bestiary of functions considered in {\cite[Corollary 22]{BlancBournezMFCS2023vantardise}}: we write $\contClasslight+\xi$ for the algebra where some real extension of function $\xi$ is added as a base function.

\begin{corollary}[A bestiary of  functions] \label{corobestiary}
	There exist
	\begin{enumerate}
		\item $\xi_1, \xi_2 : \N^2 \times \R \mapsto \R ~  \in \contClasslight+\xi$  
		such that, for all $n,m \in \N$,  $\lfloor x \rfloor \in [- 2^{n}+1, 2^{n}]$, 
		whenever $ x \in [\lfloor x \rfloor - \frac{1}{2}, \lfloor x \rfloor + \frac{1}{4}] $  , $|\xi_1(2^m, 2^n,x) - \{ x \}| \le 2^{-m}$, 
		and whenever $ x \in [\lfloor x \rfloor, \lfloor x \rfloor + \frac{3}{4}] $  , $|\xi_2(2^m, 2^n,x)-\{ x \}| \le 2^{-m}$.
		\item $\sigma_1, \sigma_2 : \N^2 \times \R \mapsto \R ~  \in \contClasslight+\xi$  
		such that, for all $n,m \in \N$,  $\lfloor x \rfloor \in [-2^{n}+1 , 2^{n}]$,
		whenever  $ x \in [\lfloor x \rfloor - \frac{1}{2}, \lfloor x \rfloor + \frac{1}{4}]$, $|\sigma_1(2^m, 2^n,x)-\lfloor x \rfloor| \le 2^{-m}$, 
		and whenever  $ x \in I_{2}=[\lfloor x \rfloor, \lfloor x \rfloor + \frac{3}{4}]$, $|\sigma_2(2^m, 2^n,x)-\lfloor x \rfloor|\le 2^{-m}$.
		\item $\lambda : \N^2 \times\R \mapsto [0,1]   \in \contClasslight+\xi$  
		such that for all $m,n\in\N$,  $\lfloor x \rfloor \in [-2^{n}+1 , 2^{n}]$, 
		whenever  $ x \in [\lfloor x \rfloor + \frac{1}{4}, \lfloor x \rfloor + \frac{1}{2}] $, $|\lambda(2^m, 2^n,x)-0| \le 2^{-m}$, 
		and whenever $  x \in [\lfloor x \rfloor + \frac{3}{4}, \lfloor x \rfloor +1] $, $|\lambda(2^m, 2^n,x)-1|\le 2^{-m}$.
		\item $\mod_{2} : \N^2 \times\R \mapsto [0,1]   \in \contClasslight+\xi$  
		such that for all $m,n\in\N$,  $\lfloor x \rfloor \in [-2^{n}+1 , 2^{n}]$, 
		whenever  $ x \in [\lfloor x \rfloor -\frac14, \lfloor x \rfloor
		+ \frac{1}{4}] $, $|\mod_{2}(2^m, 2^n, x)$-$\lfloor x \rfloor \mod 2| \le 2^{-m}$.
		\item $\div_{2} : \N^2 \times\R \mapsto [0,1]  \in \contClasslight+\xi$  
		such that for all $m,n\in\N$,  $\lfloor x \rfloor \in [-2^{n}+1 , 2^{n}]$, 
		whenever  $ x \in [\lfloor x \rfloor - \frac14, \lfloor x \rfloor +
		\frac{1}{4}] $, $|\div_{2}(2^m, 2^n, x)-\lfloor x \rfloor / / 2| \le 2^{-m}$, with $/ /$ the integer division.
	\end{enumerate}
\end{corollary}

\begin{proof}
	There were given by  $\xi_{1}(M,N,x)=\xi(M,N,x-\frac38) -\frac12$,  $\xi_{2}(M,N,x)=\xi(N,x-\frac78)$, $\sigma_i (M, N,x) = x - \xi_i(M, N,x)$,
	$\lambda(M, N,x)=\sigtanh(2M,1/4,1/2,\xi(2M, N,x-9/8))$, $\mod_{2}(M, N,x)=1-\lambda(M, N/2,\frac12x+\frac78)$,
	$\div{2}(M,N,x)=\frac12(\sigma_{1}(M,N,x)-\mod_{2}(M, N,x))$.
\end{proof}

Similarly, the equivalent of {\cite[Lemmas 23,24 and 25]{BlancBournezMFCS2023vantardise}} still hold in $\contClasslight+\xi$. Namely:

\begin{lemma} \label{tricksigmoidtanh}
There exists $\TTtanh \in \contClasslight+\xi$ such that, $l \in [0,1]$, if we take $|d'-0| \le 1/4$, 
then $|\TTtanh(2^{m},d',l)-0| \le 2^{-m}$, and if we take $|d'-1| \le 1/4$, then $|\TTtanh(2^{m},d',l)-l| \le 2^{-m}$.
\end{lemma}

\begin{lemma}\label{lem:switch:tanh}
Let $\alpha_{1},  \alpha_{2},\dots,$ $\alpha_{n}$ be some integers, and $V_{1}, V_{2}, \dots,V_{n}$ some constants. 
We write  
$\send({\alpha_{i} \sendsymbol V_{i}})_{i \in \{1,\dots,n\}}$ for the function that maps any $x \in [\alpha_{i}-1/4, \alpha_{i}+1/4]$ to $V_{i}$, for all $i \in \{1,\dots,n\}$. 

There is some  function in $\contClasslight+\xi$, that we write 
$\sendtanh(2^{m},{\alpha_{i} \sendsymbol V_{i}})_{i \in \{1,\dots,n\}}$, that maps any $x \in [\alpha_{i}-1/4, \alpha_{i}+1/4]$ to a real at distance at most $2^{-m}$ of $V_{i}$, for all $i \in \{1,\dots,n\}$. 
\end{lemma}

\begin{lemma} \label{lem:switch:pairs:tanh}
Let $N$ be some integer. Let $\alpha_{1},  \alpha_{2},\dots, \alpha_{n}$ be some integers, and   $V_{i,j}$ for $1 \le i \le n$ some constants, with $0 \le j < N$. 
We write  
$\send({(\alpha_{i},j) \sendsymbol V_{i,j}})_{i \in \{1,\dots,n\}, j \in \{0,\dots,N-1\}}$ for the function that maps any $x  \in  [\alpha_{i}-1/4, \alpha_{i}+1/4]$
 and $y \in [j-1/4,j+1/4]$ to $V_{i,j}$, for all $i \in \{1,\dots,n\}$, $j \in \{0,\dots,N-1\}$.

There is some  function in $\contClasslight+\xi$, that we write 
$\sendtanh(2^{m},{(\alpha_{i},j) \sendsymbol V_{i,j}})_{i \in \{1,\dots,n\}, j \in \{0,\dots,N-1\}}$, that maps any $x  \in  [\alpha_{i}-1/4, \alpha_{i}+1/4]$
 and $y \in [j-1/4,j+1/4]$ to a real at distance at most $2^{-m}$ of $V_{i,j}$, for all $i \in \{1,\dots,n\}$, $j \in \{0,\dots,N-1\}$.
\end{lemma}

\paragraph*{Working with one step of a Turing machine}
As the proof of {\cite[Lemmas 30]{BlancBournezMFCS2023vantardise}} is done using all the  functions provided by these lemmas, we  obtain:

\SHORTERICALP{\cite[Lemma 28]{BlancBournezMFCS2023vantardise} is done using all the  functions provided by these lemmas, we  obtain:
	
	\begin{lemma}
		We can construct some function $\overline{\Next}$ in $\contClasslight+\xi$
		that simulates one step of $M$: given a configuration $C$, writing $C'$
		for the next configuration, we have for all integer $m$,
		$\| \bar{\Next}(2^{m},\bar C) - \bar C' \| \le 2^{-m}$.
	\end{lemma}
	
	We can get  the improvement  {\cite[Lemmas 30]{BlancBournezMFCS2023vantardise}}: }

\begin{lemma} \label{robrob}
	We can construct some function $\bar {\Next}$ in $\contClasslight+\xi$ that simulates one step of $M$, i.e. that computes the $\Next$ function sending a configuration $\bar C$ of Turing machine $M$ to $\bar C'$, where $C'$ is the next one:  $\| \Next(2^{m}, 2^{S},\bar C) - \bar C' \| \le 2^{-m}$.  Furthermore, it is robust to errors on its input, up to space $S$: considering $\|\tilde{C}-\bar C\| \le 4^{-(S+2)}$, $\| \Next(2^{m}, 2^{S},\tilde C) - \bar C' \| \le 2^{-m}$ remains true.
\end{lemma}

\paragraph*{Converting integers an dyadics to words and conversely}

The authors of \cite{BlancBournezMFCS2023vantardise} also defined some functions for converting integers and dyadics to their encoding as words, and conversely. Namely, they consider the following encoding: every digit in the binary expansion of dyadic $d$  is encoded by a pair of symbols in the radix $4$ expansion of $\overline{d} \in \Image \cap [0,1]$: digit $0$ (respectively: $1$) is encoded by $11$ (resp. $13$) if before the ``decimal'' point in $d$, and digit $0$ (respectively: $1$) is encoded by $31$ (resp. $33$) if after. For example, for $d=101.1$ in base $2$, $\overline{d}=0.13111333$ in base $4$. Conversely, given $\overline{d}$, they provided a way to construct $d$.
This corresponds to  {\cite[Lemmas 33 and 34]{BlancBournezMFCS2023vantardise}}:

\begin{lemma}[{From $\N$ to $\Image$}]  \label{lem:manquant} We can construct some function $\Decode: \N^{2} \to \R$ in $\manonclasslighttanh$ that maps $m$ and $n$ to some point at distance less than $2^{-m}$ from $\encodagemot(\overline{n})$.
\end{lemma}

\begin{lemma}[{From $\Image$ to $\R$}, and multiplying in parallel]  \label{lem:codage:manon}
	We can construct some  function $\EncodeMul: \N^{2} \times [0,1] \times \R \to \R$ in $\manonclasslighttanh$ that  maps $m$, $2^{S}$, $\encodagemot(\overline{d})$  and (bounded) $\lambda$ to some real at distance at most $2^{-m}$ from  $\lambda d$, whenever  $\overline{d}$ is of length less than $S$.
\end{lemma}

As for $\xi$, it is not clear that we have some real extensions of these functions in  $\contClasslight$: we write $\contClasslight+\xi+\Decode+\Encode$ for the algebra where some real extension of these functions is added as a base function.

\subsection{Constructing the missing functions}

We need a way to construct some substitute of ``missing functions'' ($\xi$, $\Decode$ and $\EncodeMul$). As all of them are defined using discrete ODEs, an idea is to use a continuous ODE to simulate the respective discrete ODEs: we hence revisit the construction of the ideal iteration trick of Section \ref{ideal:branicky}, dealing with errors and not exact functions $\theta(z)$ and $r(x)$. 

The key is to revisit Lemma \ref{bran:funda}, and do a more detailed analysis of possible involved errors in dynamics of the form \eqref{eq:de:branicky}. This equation has been studied by various authors in several articles, including \cite{braverman2005hyperbolic,CMC00,dsg05, JournalOfComplexity2016,TheseRiccardo}. We use the following statement from \cite[Lemma 4.5]{TheseRiccardo}, \cite[Lemma 5.2]{BGDPRiccardo2022}, obtained basically  by a case analysis of error propagations in Lemma  \ref{bran:funda}. 

\begin{lemma}[{Improved error analysis of ODE \eqref{eq:de:branicky},   \cite[Lemma 4.5]{TheseRiccardo} \cite[5.2]{BGDPRiccardo2022}}]  \label{lem:quatre:cinq}
	Consider a point $b \in \R$, some $\gamma>0$ some reals $t_0<t_1$, and a function $\phi: \R \rightarrow \R$ with the property that $\phi(t) \geq 0$ for all $t \geq t_0$ and $\int_{t_0}^{t_1} \phi(t) d t>0$. Let $\rho, \delta \geq 0$ and let $\bar{b}, E: \mathbb{R} \rightarrow \mathbb{R}$ be functions such that that $|\bar{b}(t)-b| \leq \rho$ and $|E(t)| \leq \delta$ for all $t \geq t_0$. Then the IVP defined by
	$$
	z^{\prime}=c(\bar{b}(t)-z)^3 \phi(t)+E(t)
	$$
	with the initial condition $z\left(t_0\right)=\bar{z}_0$, 
	where $\gamma>0$ and $c \ge \frac{1}{2 \gamma^2 \int_{t_0}^{t_1} \phi(t) dt}$ satisfies
	\begin{enumerate}
		\item  $\left|z\left(t_1\right)-b\right|<\rho+\gamma+\delta\left(t_1-t_0\right)$, independently of the initial condition $\bar{z}_0 \in \mathbb{R}$
		\item $\min \left(\bar{z}_0, b-\rho\right)-\delta\left(t_1-t_0\right) \leq z(t) \leq \max \left(\bar{z}_0, b+\rho\right)+\delta\left(t_1-t_0\right)$ for all $t \in\left[t_0, t_1\right]$.
	\end{enumerate}
\end{lemma}

\begin{proposition}[Simulating a discrete ODE by a continuous ODE] \label{the:manon:trick}
Assumge $\tu G$ is almost constant around  $\N\delta$ and $r$ is a rounding function around $\N\delta$  for some $\delta>0$. 
	Suppose that, in \eqref{code}, we replace function $\theta(z)$ and function $r(z)$ by some suitable approximations: we take $\theta(x)=\relu(x)$, $\theta_{\epsilon'}(x)$, $r_{\epsilon'(z)}$ such that $\theta(z)=_{\epsilon'}\theta(z)$, and $r_{\epsilon'}(x)=_{\epsilon'}$,  and take constant $c$ big enough. 	Then the solution of the obtained ODE will continuously simulate the discrete ODE \eqref{dode}, with same bounds as in the analysis in Section \ref{ideal:branicky}, i.e. with error at most $\epsilon$ if $\epsilon'$ is taken sufficiently small. To guarantee $\epsilon=2^{-n}$, it is sufficient to take $\epsilon'=2^{-p(n)}$  and $\theta_{\epsilon'}(x)=\relutanh(2^{p(n)},x)$ for some polynomial $p$. 
\end{proposition}

\begin{proof}
	The key is that  involved errors propagate additively,  from Lemma \ref{lem:quatre:cinq}. Namely, they are in $\mathcal{O}(\epsilon')$, but they are then corrected from the reasoning in Section \ref{ideal:branicky}: rounding function corrects errors or order $\epsilon$ whenever its argument is at distance less than $1/4\delta$ of some $n\delta$ exactly as in the reasoning in Section \ref{ideal:branicky} (where $\delta=1$, even if now it introduces some error $\epsilon'$ at every step; but the latter is corrected at the next step). Observe that the involved constant $c$, is of order $2^{n}$. 
	
	We did up to that point the reasoning, assuming that parameter $\tu x$ is fixed. But if we consider a function that either depends in a controlled way on $t$ on that parameter,  or that is such that a small perturbation on its inputs does not change much its output (we mean in a controlled way, in the way we consider the rounding function $r$), then the analysis remains perfectly valid, even when this parameter may not be exact. 
\end{proof}

We claim that for all $n\in\N$, $\tu y_1(n,\tu x) =_\epsilon \tu y_2(n,\tu x) =_\epsilon \tu f(n,\tu x)$, and $\tu y_1(t+\frac12,\tu x) =_\epsilon \tu y_2(t,\tu x) =_\epsilon \tu f(n,\tu x)$ for all $t \in [n,n+\frac12]$.

For $n=0$, initially $\tu f(0,\tu x) = \tu y_{1}(0,\tu x) = \tu y_{2}(0,\tu x) = \tu g(\tu x)$.
For $t \in [n,n+1/2]$, we have $\theta(- \sin(2\pi t))=_{\epsilon'} 0$, and hence 
$\tu y_2' =_{\epsilon'} 0$, so $\tu y_2$ is  kept close to  value $\tu g(\tu x)$ for $t \in[0,\frac{1}{2}]$, with an error less than $\frac12 \epsilon'$.

Consequently, for $t \in [0,1/2]$, $r(\tu y_{2})$ is kept close to a constant value $\tu g(\tu x)$, when an error less than $\epsilon'$, if we choose $\epsilon' < \frac14 \delta$. Meanwhile, $r(t)$ is also at a value close to $n$ with error less than $\epsilon'$.

Consequently, on this interval, 
if we write $C(t) = c\theta(\sin(2\pi t))$, then the dynamics of $\tu y_1$ is given by a dynamic of the form of Lemma \ref{lem:quatre:cinq}.
This lemma states that $\tu y_1(t,\tu x)$ will approach $\tu G(\tu g(\tu x),0,\tu x)=\tu f(1,\tu x)$ on this interval, with an error of order
$\epsilon' + \epsilon' + \frac12 \epsilon'$. 

Here the hypothesis that $\tu G$ is almost constant around  $\N\delta$ means that its value is guaranteed to be at $\epsilon'$ from $\tu G(\tu g(\tu x),0,\tu x)$ on the interval. 

Thus, $\tu y_1(\frac{1}{2},\tu x) =_{\epsilon/2}  \tu f(1,\tu x)$, if we choose $\frac52 \epsilon' < \epsilon/2$. 
At $t=n+\frac{1}{2}$, $\tu y_1$ will hence have simulated one step of discrete ODE \eqref{dode}, with error less than $\epsilon/2$, 
and $\tu y_{2}$ will be close to $ \tu g(\tu x)$ with error less than $\epsilon' < \epsilon/2$.

Now, for $t \in [n+\frac{1}{2}, n+1]$ the roles of $\tu y_1$ and $\tu y_2$ are exchanged : $\tu y_1'(t, \tu x) =_{\epsilon'} 0$, so $\tu y_1$ is kept almost fixed, with a new error less than $\frac12 \epsilon'$.  In the same time $\tu y_2$ approaches $r(\tu y_{1})=\tu f(1, \tu x)$ by Lemma \ref{lem:quatre:cinq}, with some new error
of order less than $\frac52 \epsilon' < \epsilon/2$. 

Consequently, we get the property at rank $n+1$.

\begin{remark} Observe that, somehow, the constructions always replace every function with a function that does not change much locally (i.e., change in a controlled way). This is the key that provides a robust ODE as in Definition \ref{defRobODE}, leading to polynomial space complexity by Theorem \ref{main-direction-one}.
\end{remark}

In other words, whenever we have some discrete ODE as in \eqref{dode} defining some function $\tu f(t,\tu x)$,  we can construct some continuous ODE, using only functions from $\contClasslight$, such that one of its projection provides a function $\tu f(z, t, \tu x)$, with the guarantee $\tu f(2^{n}, t, \tu x)$ is $2^{-n}$ close to  $\tu f(n,\tu x)$, whenever $t$ is close (at distance less than $1/4$)  to some integer $n$. 

This works, as we can obtain such a $r_{\epsilon'}(x)$ from the functions from Corollary \ref{corobestiary}: Consider 
$r(x,2^{m})=\sigma_{2}(2^{m},2^{n},x+\frac14)$ that works over $\lfloor x \rfloor \in [-2^{n}+1 , 2^{n}]$, and observe that this is sufficient to apply the trick for the required functions, from the form of the considered discrete ODE in \cite{BlancBournezMFCS2023vantardise}. 

Except that we have a bootstrap problem: $\xi$ was defined using a discrete ODE in \cite{BlancBournezMFCS2023vantardise}, and as the functions from Corollary \ref{corobestiary} are defined above using $\xi$, we cannot apply this reasoning to get function $\xi$.  But the point is that for the special case of $\xi$, it is easy to construct a function in $\contClass$ that corresponds to some real extension of $\xi$, as we have functions such as $\sin(x)=\cos(\frac \pi2-x)$ and $\pi$. 
\begin{lemma} \label{lem:torture}
Function $\xi$ has some real extension in $\contClasslight$.
\end{lemma}
\begin{proof}
If we succeed to obtain a function $i(2^{m},2^{n},x)$ that values $\lfloor x \rfloor$ whenever $x \in [\lfloor x \rfloor, \lfloor x \rfloor + \frac{3}{4}]$, we are done, as we can then obtain
$\xi(2^{m},2^{n},x)$ by considering  $\xi(2^{m},2^{n},x)=x+\frac78-i(2^{m},2^{n},x+\frac78)$. 

A possible solution is then the following: consider function $R_{e}(x):=\sig(x,0,e/2)$, and then $t_{e}(x)=(1-R_{e}(\sin(2\pi x)))((1-R_{e}(\sin(4\pi x)))$. If we put aside some interval of width $e/2$ around $\frac12$ and $\frac78$ where it takes values in $[0,1]$,  it values $0$ on $[\lfloor x \rfloor,\lfloor x \rfloor+\frac78]$, and then $1$ on $[\lfloor x \rfloor+\frac78,\lfloor x \rfloor+1]$ (see following graphical illustration). We can then consider $I_{e}(t)= 8 \int_{0}^{t} t_{e}(x) dx$ (i.e. the solution of ODE $l'_{e}= 8 t_{e}$), and then $i(t)=_{e.t} l_{e}(t)$. 
%
%
%
It is then sufficient to replace $\sig$ by $\sigtanh$, in the above expressions, in order to  control the error and make it smaller than $2^{-m}$. 

Here is a graphical representation of $R_{\frac1{10}}(x)$:

\begin{center}
	\includegraphics[width=0.5\textwidth]{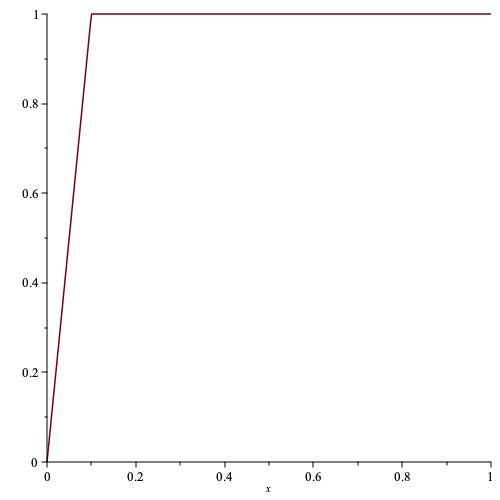}
\end{center}

Then of $R_{\frac1{10}}(\sin(2\pi x))$:

\begin{center}
	\includegraphics[width=0.5\textwidth]{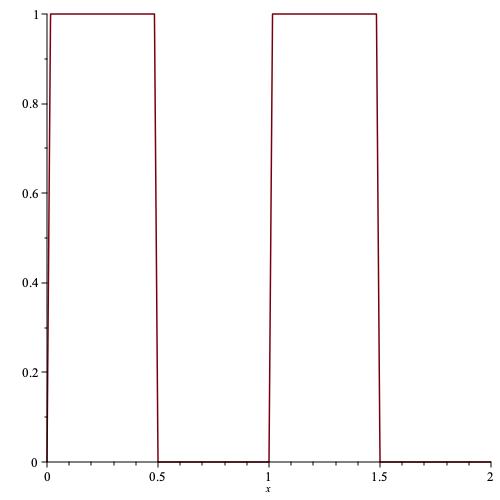}
\end{center}

and $R_{\frac1{10}}(\sin(4\pi x))$:

\begin{center}
	\includegraphics[width=0.5\textwidth]{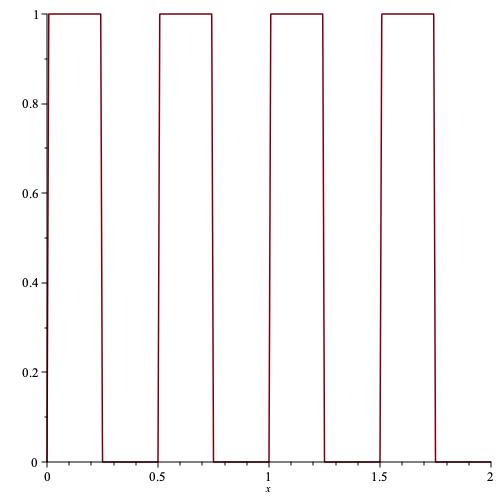}
\end{center}

We then get $t_{\frac1{10}}(x)$.

\begin{center}
	\includegraphics[width=0.5\textwidth]{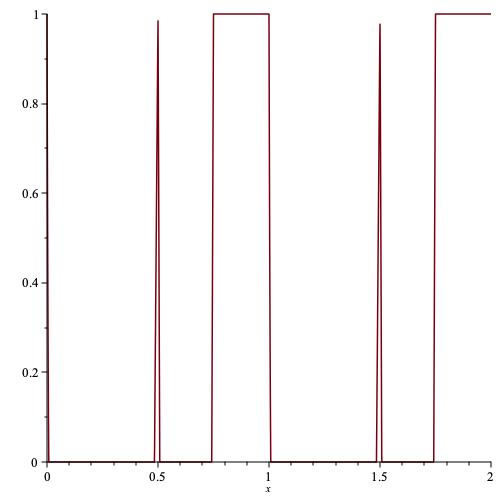}
\end{center}

Its integral is then close to $\frac18 \lfloor x \rfloor$ on $[\lfloor x \rfloor, \lfloor x \rfloor + \frac{3}{4}]$.

\end{proof}
%

Consequently, this is true that we can substitute a discrete ODE with a continuous ODE for the required functions $\Decode$ and $\EncodeMul$: just replace $\xi$ in the involved schemas by the above function.  Notice that we can also easily get a real extension of the function that maps $n$ to $2^{n}$.

\subsection{Working with all steps of a Turing machines}

We can then go from one step of a Turing machine, to arbitrarily many steps. We are following the idea of \cite{BlancBournezMFCS2023vantardise}, but replacing discrete ODEs with continuous ODEs.

\begin{theorem} \label{th:trendeux:new} 
	Consider some Turing machine $M$ that computes some function $f: \Sigma^{*} \to \Sigma^{*}$ in some polynomial space $S(\ell(\omega))$ on input $\omega$.  One can construct some function $\tilde{ f}: \N^{2} \times \R \to \R$ in $\contClass$ that does the same: we have $\tilde{ f}(2^{m}, 2^{S(\ell(\omega))},
	\encodagemot(\omega))$ that is at most $2^{-m}$ far from $\encodagemot(f(\omega))$.
\end{theorem}

\begin{proof}
	We denote by $\mathcal{M}$ the Turing machine computing $f$. Similarly to the arguments in \cite{BlancBournezMFCS2023vantardise}, we can state that there exists a function $Exec$ solution of a robust linear discrete ODE (E) that "computes" the execution of $\mathcal{M}$, with $C_{init}$ the initial configuration : 
	\begin{equation*}
		(E):
		\left\{
		\begin{aligned}
			&Exec(2^m, 0, 2^S,C_{init}) = C_{init} \\
			&\frac{\delta Exec(2^m, t, 2^S,C_{init})}{\delta t} = Next(2^m, 2^S, Exec(2^m, t, 2^S, C_{init})) - Exec(2^m, t, 2^S, C_{init})
		\end{aligned}
		\right.
	\end{equation*}
	
	For any configuration $\overline{C}$ of $\mathcal{M}$, let write 	$F(\overline{C}) )=F(2^m, 2^S,\overline{C}) = Next(2^m, 2^S, \overline{C}) + \overline{C}$, associated to the righthand side of the above discrete ODE. 	Denoting by $\tilde{C}$ the errorless encoding of the configuration $C$,  from the constructions of \cite{BlancBournezMFCS2023vantardise} (Lemma \ref{robrob}), it is true that 
	if $\left| \overline{C} - \tilde{C} \right| \leq 4^{-(S+2)} $, then  $\left|F(\overline{C})-F(\tilde{C})\right| \leq 4^{-(S+2)}$. 	$F$ does not change much locally on the space of configuration. Denoting by $S$ the space of $\mathcal{M}$,
	and replacing $m$ by $m + 2S + 4$ as in \cite{BlancBournezMFCS2023vantardise}, we have
	$\left|Next(2^m, 2^S, \overline{C}) - \overline{C} \right| \leq 4^{-(S+2)}$. So at each step of the TM,
	the error is fixed (and bounded).
	We can then apply the above arguments (Proposition \ref{the:manon:trick})  to simulate continuously (E), with some controlled error: all involved quantities have encoding polynomials in the size of the inputs. 
\end{proof}

\section{Proof of Theorem \ref{th:mainone}}
\label{sec:cstm}

\begin{proof}
\noindent $\subseteq$:
	In this direction, we just need to prove that $\contClass$ contains only functions over the reals that are computable in polynomial space.  Indeed, then for a function $\tu f: \R^{d} \to \R^{d'}$ sending every integer $\vn \in \N^{d}$ to the vicinity of some integer of  $\N^{d}$,  at a distance less than $1/4$, by approximating its value with precision $1/4$ on its input arguments, and taking the closest integer, we will get a function from the integers to the integers, that corrresponds to $\DP(f)$, and that will be in $\FPspace  \cap \N^{\N}$.
	
	This is indeed the case, since i) all the base functions of $\contClasslim$ are in $\FPspace$: they are even in $\FPtime$, see \cite{Ko91} ii)  $\R^\R \cap \FPspace$ is stable under $\composition$. 
	iii)  stability under $\robODE$ follows from Theorem \ref{main-direction-one}. 

%

\noindent  $\supseteq$: 
	In the other direction, we use an argument similar to \cite{BlancBournezMFCS2023vantardise}: namely, as the function is polynomial space computable, this means that there is a polynomial space computable function $g: \N^{d''+1} \to \{\symboleun,\symboledeux\}^{*}$ so that on $\tu m,2^{n}$, it provides the encoding $\bar{\phi(\tu m,n)}$ of some dyadic  $\phi(\tu m,n)$ with $\|\phi(\tu m,n)-\tu f(\tu m)\| \le 2^{-n}$ for all $\tu m$.  
	The problem is then to decode, compute and encode the result to produce this dyadic. 
	More precisely, from Theorem \ref{th:trendeux:new},  we get $\tilde{g}$ with  
	$$|\tilde{g}(2^{e},2^{p(max(\tu m,n))}, \Decode(2^{e},\tu m,n)) -\encodagemot(g(\tu m,n)) | \le 2^{-e}$$ for some polynomial $p$ 
	corresponding to the time required to compute  $g$, and $e=\max(p(max(\tu m,n)),n)$. Then we  need to transform the  value to the correct dyadic: we mean \\
	$\tilde{\tu f}(\tu m,n)=\EncodeMul(2^{e},2^{t},\tilde{g}(2^{e},2^{t},\Decode(2^{e},\tu m,n)),1)$, where \\$t=p(max(\tu m,n))$,   $e=\max(p(max(\tu m,n)),n)$ provides a solution with 
	$\|\tilde{\tu f}(\tu m,2^{n})-\tu f(\tu m)\| \le 2^{-n}.$
\end{proof}

\section{Proof of Theorem \ref{th:main:twopSpace} } \label{sectionMain}

\begin{proof}
\noindent  $\subseteq$:
	To prove that $\contClasslim \subseteq \R^\R \cap \FPspace$, we only need to add to the previous arguments that 
	$\R^\R \cap \FPspace$ is also stable under $\MANONlim$.
%
%
%
%

\noindent   $\supseteq$:  In this direction, we have the same issue as in \cite{BlancBournezMFCS2023vantardise}: the strategy of decoding, working with the Turing machine, and encoding is not guaranteed to work for all inputs. But, we can solve it by using an adaptative barycenter technique as in \cite{BlancBournezMFCS2023vantardise}. 

We recall the principle here for a function whose domain is $\R$, but it can be generalised to $\R^d$.
	The idea is to construct some function $\lambda:\N^2\times\R \rightarrow [0,1]$ definable in $\contClasslight $ as in Corollary \ref{corobestiary}, but with a continuous ODE : 
	Adapting the proof from \cite{BlancBournezMFCS2023vantardise} and using the simulation of $\xi$ in our continuous framework, we can consider  $\lambda(2^m, N,x)= \varPsi(\Xi(2^{m+1}, N,x-9/8))$ where $\varPsi(x)=\sigtanh(2^{m+1},1/4,1/2,x)$. In particular, by definition, $\lambda \in \contClasslight$.
	Thus, by Lemma \ref{thXi}, if $\lambda(2^m, N, x) =_{2^{-m}} 0$, then 
	$\sigma_2(2^m, N,x) =_{2^{-m}} \lfloor x \rfloor $. If $\lambda(2^m, N, x) =_{2^{-m}} 1$, then 
	$\sigma_1(2^m, N,x) =_{2^{-m}} \lfloor x \rfloor $ and if $\lambda(2^m, N, x)\in (0,1)$, 
	then $\sigma_1(2^m, N,x) =_{2^{-m}} \lfloor x \rfloor +1 $ and furthermore  $\sigma_2(2^m, N,x) =_{2^{-m}} \lfloor x \rfloor $. 
	So, $$\lambda(\cdot, 2^n,x) \textit{Formula}_{1}(x,u,M,n) + (1-\lambda(\cdot, 2^n,n)) \textit{Formula}_{2}(x,u,M,n) $$ and we are
	sure to be close (up to some bounded error) to some $2^{-m}$ approximation of a function $f$.
\end{proof}

%
%
	
	
	

\section{Conclusion} \label{sec:conclusion} 

We characterised polynomial space using an algebraically defined class of functions, using a finite set of basic functions, closure under composition, and a schema for defining functions from robust ODEs. We proposed a concept of robust ODEs solvable in polynomial space. As far as we know, this is an original method for solving ODEs optimising space. It is based on classical constructions such as Savitch's theorem.  We extended existing characterisations to a characterisation of functions over the reals and not only over the integers.

The interesting message from our statements is that we provide a clear and simple concept associated with continuous ODEs for space: space corresponds to the precision for numerically stable systems. Hence, compiled with \cite{JournalACM2017}, we now know the length of solutions corresponds to time and precision to memory.  

Considering future work: 
We have an algebraically defined class of functions. It remains to know whether this could be transferred at the level of polynomial ODE. We know that soluting of polynomial ODEs defined a very robust class of functions, stable by many operations: sum, products, division, ODE solving, etc: see \cite{TheseDaniel,InformationAndComputation2017}. Hence, all the base functions we consider in our algebraic class can be turned into polynomial ODEs, by adding some variables. 
It would be interesting to understand if we could define space complexity directly at the level of polynomial ODEs, using precision. 

Recently, a characterisation of $\Pspace$ was obtained for polynomial ODEs using rather ad-hoc definitions in \cite{TheseRiccardo,BGDPRiccardo2022} and working over a non-compact space. Could our characterisation be put at this simplest class of ODEs, but working with precision? The point is that the characterisation there uses unbounded domains, hence, precision is harder to interpret in their constructions, where the schemas are somehow done to control errors. 

Of course, from our statements, adding any $\FPspace$-computable function over the reals among the base functions would not change the class. However, we did not intend to minimise the number of base functions. For example, $\tanh(t)$ is solution of ODE $f'=1+f^{2}$ and $\cos(t)$ can be obtained by the two dimensional ODE $y_{1}'=-y_{2}$, $y_{2}'=y_{1}$.  Minimising the number of base functions is also left for future work. We believe that even in this settings, proving space complexity corresponds to precision is already significant, independently of this question of a minimal set of base functions. 

\newpage
\bibliographystyle{plainurl}

\bibliography{./bournez,./perso}

\begin{thebibliography}{10}

\bibitem{Abe71}
Oliver Aberth.
\newblock The failure in computable analysis of a classical existence theorem
  for differential equations.
\newblock {\em Proceedings of the American Mathematical Society}, 30:151--156,
  1971.

\bibitem{bertschinger10training}
Daniel Bertschinger, Christoph Hertrich, Paul Jungeblut, Tillmann Miltzow, and
  Simon Weber.
\newblock {Training Fully Connected Neural Networks is $\exists
  \mathbb{R}$-Complete}.
\newblock {\em Preprint, https://doi. org/10.48550/arXiv}, 2204, 2022.

\bibitem{BlancBournezMCU22vantardise}
Manon Blanc and Olivier Bournez.
\newblock A characterization of polynomial time computable functions from the
  integers to the reals using discrete ordinary differential equations.
\newblock In J{\'{e}}r{\^{o}}me Durand{-}Lose and Gy{\"{o}}rgy Vaszil, editors,
  {\em Machines, Computations, and Universality - 9th International Conference,
  {MCU} 2022, Debrecen, Hungary, August 31 - September 2, 2022, Proceedings},
  volume 13419 of {\em Lecture Notes in Computer Science}, pages 58--74.
  Springer, 2022.
\newblock {MCU'22 Best Student Paper Award}.
\newblock \href {https://doi.org/10.1007/978-3-031-13502-6\_4}
  {\path{doi:10.1007/978-3-031-13502-6\_4}}.

\bibitem{BlancBournezMFCS2023vantardise}
Manon Blanc and Olivier Bournez.
\newblock A characterisation of functions computable in polynomial time and
  space over the reals with discrete ordinary differential equations:
  Simulation of turing machines with analytic discrete odes ({MFCS'2023} best
  paper award).
\newblock In {\em 48th International Symposium on Mathematical Foundations of
  Computer Science (MFCS 2023)}, volume 272. Schloss Dagstuhl-Leibniz-Zentrum
  f{\"u}r Informatik, 2023.

\bibitem{BlancBournezMFCS2023Journal}
Manon Blanc and Olivier Bournez.
\newblock Simulation of turing machines with analytic discrete {ODEs}: {FPTIME}
  and {FPSPACE} over the reals characterised with discrete ordinary
  differential equations.
\newblock {\em arXiv preprint arXiv:2307.11747}, 2023.

\bibitem{csl24}
Manon Blanc and Olivier Bournez.
\newblock Quantifiying the robustness of dynamical systems. relating time and
  space to length and precision.
\newblock In {\em Computer Science Logic CSL'24}, Naples, Italy, February 2024.

\bibitem{BCSS98}
Lenore Blum, Felipe Cucker, Mike Shub, and Steve Smale.
\newblock {\em Complexity and Real Computation}.
\newblock Springer, 1998.

\bibitem{BSS89}
Lenore Blum, Mike Shub, and Steve Smale.
\newblock On a theory of computation and complexity over the real numbers; {NP}
  completeness, recursive functions and universal machines.
\newblock {\em Bulletin of the American Mathematical Society}, 21(1):1--46,
  July 1989.

\bibitem{MFCS2019}
Olivier Bournez and Arnaud Durand.
\newblock Recursion schemes, discrete differential equations and
  characterization of polynomial time computation.
\newblock In Peter Rossmanith, Pinar Heggernes, and Joost{-}Pieter Katoen,
  editors, {\em 44th Int Symposium on Mathematical Foundations of Computer
  Science, {MFCS}}, volume 138 of {\em LIPIcs}, pages 23:1--23:14. Schloss
  Dagstuhl - Leibniz-Zentrum f{\"{u}}r Informatik, 2019.

\bibitem{MFCSJournal}
Olivier Bournez and Arnaud Durand.
\newblock A characterization of functions over the integers computable in
  polynomial time using discrete ordinary differential equations.
\newblock {\em Computational Complexity}, 32(2):7, 2023.

\bibitem{BGDPRiccardo2022}
Olivier Bournez, Riccardo Gozzi, Daniel~S Gra{\c{c}}a, and Amaury Pouly.
\newblock A continuous characterization of {PSPACE} using polynomial ordinary
  differential equations.
\newblock {\em Journal of Complexity}, 77:101755, august 2023.
\newblock URL:
  \url{https://www.sciencedirect.com/science/article/pii/S0885064X23000249?dgcid)=author}.

\bibitem{InformationAndComputation2017}
Olivier {Bournez}, Daniel {Gra{\c c}a}, and Amaury {Pouly}.
\newblock {On the Functions Generated by the General Purpose Analog Computer}.
\newblock {\em Information and Computation}, 257:34--57, 2017.
\newblock \href {http://arxiv.org/abs/1602.00546} {\path{arXiv:1602.00546}},
  \href {https://doi.org/10.1016/j.ic.2017.09.015}
  {\path{doi:10.1016/j.ic.2017.09.015}}.

\bibitem{JournalOfComplexity2016}
Olivier Bournez, Daniel Gra{\c c}a, and Amaury Pouly.
\newblock Computing with polynomial ordinary differential equations.
\newblock {\em Journal of Complexity}, 36:106 -- 140, 2016.
\newblock URL:
  \url{http://www.sciencedirect.com/science/article/pii/S0885064X16300255},
  \href {https://doi.org/http://dx.doi.org/10.1016/j.jco.2016.05.002}
  {\path{doi:http://dx.doi.org/10.1016/j.jco.2016.05.002}}.

\bibitem{bournez2012complexity}
Olivier Bournez, Daniel~S Gra{\c{c}}a, and Amaury Pouly.
\newblock On the complexity of solving initial value problems.
\newblock In {\em Proceedings of the 37th International Symposium on Symbolic
  and Algebraic Computation}, pages 115--121, 2012.

\bibitem{ICALP2016vantardise}
Olivier Bournez, Daniel~S. Gra{\c{c}}a, and Amaury Pouly.
\newblock {Polynomial Time corresponds to Solutions of Polynomial Ordinary
  Differential Equations of Polynomial Length. The General Purpose Analog
  Computer and Computable Analysis are two efficiently equivalent models of
  computations ({ICALP'2017} Track B best paper award)}.
\newblock In {\em 43rd International Colloquium on Automata, Languages, and
  Programming, {ICALP} 2016, July 11-15, 2016, Rome, Italy}, volume~55 of {\em
  LIPIcs}, pages 109:1--109:15. Schloss Dagstuhl - Leibniz-Zentrum fuer
  Informatik, 2016.

\bibitem{JournalACM2017}
Olivier Bournez, Daniel~S. Gra{\c c}a, and Amaury Pouly.
\newblock {Polynomial Time corresponds to Solutions of Polynomial Ordinary
  Differential Equations of Polynomial Length}.
\newblock {\em Journal of the ACM}, 64(6):38:1--38:76, 2017.
\newblock \href {https://doi.org/10.1145/3127496} {\path{doi:10.1145/3127496}}.

\bibitem{ICALP2017}
Olivier Bournez and Amaury Pouly.
\newblock A universal ordinary differential equation.
\newblock In {\em International Colloquium on Automata Language Programming,
  ICALP'2017}, 2017.

\bibitem{bournez2021survey}
Olivier Bournez and Amaury Pouly.
\newblock A survey on analog models of computation.
\newblock In {\em Handbook of Computability and Complexity in Analysis}, pages
  173--226. Springer, 2021.

\bibitem{Bra95}
M.~S. Branicky.
\newblock Universal computation and other capabilities of hybrid and continuous
  dynamical systems.
\newblock {\em Theoretical Computer Science}, 138(1):67--100, 6~February 1995.

\bibitem{brattka2008tutorial}
Vasco Brattka, Peter Hertling, and Klaus Weihrauch.
\newblock A tutorial on computable analysis.
\newblock In {\em New computational paradigms}, pages 425--491. Springer, 2008.

\bibitem{braverman2005hyperbolic}
Mark Braverman.
\newblock Hyperbolic {Julia} sets are poly-time computable.
\newblock {\em Electronic Notes in Theoretical Computer Science}, 120:17--30,
  2005.

\bibitem{CMC00}
Manuel~L. Campagnolo, Cristopher Moore, and Jos{\'e}~F{\'e}lix Costa.
\newblock Iteration, inequalities, and differentiability in analog computers.
\newblock {\em Journal of Complexity}, 16(4):642--660, 2000.

\bibitem{chen2018neural}
Tian~Qi Chen, Yulia Rubanova, Jesse Bettencourt, and David~K Duvenaud.
\newblock Neural ordinary differential equations.
\newblock In {\em Advances in Neural Information Processing Systems}, pages
  6571--6583, 2018.

\bibitem{Clo95}
P.~Clote.
\newblock Computational models and function algebras.
\newblock In Edward~R. Griffor, editor, {\em Handbook of Computability Theory},
  pages 589--681. North-Holland, Amsterdam, 1998.

\bibitem{clote2013boolean}
Peter Clote and Evangelos Kranakis.
\newblock {\em Boolean functions and computation models}.
\newblock Springer Science \& Business Media, 2013.

\bibitem{Cob65}
Alan Cobham.
\newblock The intrinsic computational difficulty of functions.
\newblock In Y.~Bar-Hillel, editor, {\em Proceedings of the International
  Conference on Logic, Methodology, and Philosophy of Science}, pages 24--30.
  North-Holland, Amsterdam, 1962.

\bibitem{collins2008effectivesimpl}
Pieter Collins and Daniel~S Gra{\c{c}}a.
\newblock Effective computability of solutions of ordinary differential
  equations the thousand monkeys approach.
\newblock {\em Electronic Notes in Theoretical Computer Science}, 221:103--114,
  2008.

\bibitem{collins2009effective}
Pieter Collins and Daniel~S. Gra{\c{c}}a.
\newblock {Effective Computability of Solutions of Differential Inclusions The
  Ten Thousand Monkeys Approach}.
\newblock {\em Journal of Universal Computer Science}, 15(6):1162--1185, 2009.

\bibitem{Dem96}
J.-P. Demailly.
\newblock {\em Analyse Num\'{e}rique et \'{E}quations Diff\'{e}rentielles}.
\newblock Presses Universitaires de Grenoble, 1996.

\bibitem{Dev89a}
R.~L. Devaney.
\newblock {\em An Introduction to Chaotic Dynamical Systems}.
\newblock Addison-Wesley, 2nd edition, 1989.

\bibitem{etessami2010complexity}
Kousha Etessami and Mihalis Yannakakis.
\newblock On the complexity of {Nash} equilibria and other fixed points.
\newblock {\em SIAM Journal on Computing}, 39(6):2531--2597, 2010.

\bibitem{CMSB17vantardise}
Francois Fages, Guillaume Le~Guludec, Olivier Bournez, and Amaury Pouly.
\newblock Strong turing completeness of continuous chemical reaction networks
  and compilation of mixed analog-digital programs.
\newblock In {\em Computational Methods in Systems Biology-CMSB 2017}, 2017.
\newblock {CMSB'2017 Best Paper Award}.

\bibitem{TheseRiccardo}
Riccardo Gozzi.
\newblock {\em Analog Characterization of Complexity Classes}.
\newblock PhD thesis, Instituto Superior T{\'e}cnico, Lisbon, Portugal and
  University of Algarve, Faro, Portugal, 2022.

\bibitem{TheseDaniel}
Daniel~S. Gra{\c c}a.
\newblock {\em Computability with Polynomial Differential Equations}.
\newblock PhD thesis, Instituto Superior T{\'e}cnico, 2007.

\bibitem{dsg05}
Daniel~S. Gra{\c c}a, Manuel~L. Campagnolo, and Jorge Buescu.
\newblock Robust simulations of {T}uring machines with analytic maps and flows.
\newblock In B.~Cooper, B.~Loewe, and L.~Torenvliet, editors, {\em Proceedings
  of CiE'05, New Computational Paradigms}, volume 3526 of {\em Lecture Notes in
  Computer Science}, pages 169--179. Springer-Verlag, 2005.

\bibitem{GraCos03}
Daniel~S. Gra{\c c}a and Jos{\'e}~F{\'e}lix Costa.
\newblock Analog computers and recursive functions over the reals.
\newblock 19(5):644--664, 2003.

\bibitem{dsg06a}
Daniel~S. Gra{\c c}a, N.~Zhong, and J.~Buescu.
\newblock Computability, noncomputability and undecidability of maximal
  intervals of {IVP}s.
\newblock {\em Transactions of the American Mathematical Society}, 2006.
\newblock To appear.

\bibitem{GracaZhongHandbook}
Daniel~S. Gra{\c c}a and Ning Zhong.
\newblock {\em Handbook of Computability and Complexity in Analysis}, chapter
  Computability of Differential Equations.
\newblock Springer., 2018.

\bibitem{gracca2023robust}
Daniel~S Gra{\c{c}}a and Ning Zhong.
\newblock Robust non-computability and stability of dynamical systems.
\newblock {\em arXiv preprint arXiv:2305.14448}, 2023.

\bibitem{HSD03}
Morris~W. Hirsch, Stephen Smale, and Robert Devaney.
\newblock {\em Differential Equations, Dynamical Systems, and an Introduction
  to Chaos}.
\newblock Elsevier Academic Press, 2003.

\bibitem{kawamura2009lipschitz}
Akitoshi Kawamura.
\newblock {Lipschitz continuous ordinary differential equations are
  polynomial-space complete}.
\newblock In {\em 2009 24th Annual IEEE Conference on Computational
  Complexity}, pages 149--160. IEEE, 2009.

\bibitem{kawamura2014computational}
Akitoshi Kawamura, Hiroyuki Ota, Carsten R{\"o}snick, and Martin Ziegler.
\newblock Computational complexity of smooth differential equations.
\newblock {\em Logical Methods in Computer Science}, 10, 2014.

\bibitem{kawamura2018parameterized}
Akitoshi Kawamura, Florian Steinberg, and Holger Thies.
\newblock Parameterized complexity for uniform operators on multidimensional
  analytic functions and ode solving.
\newblock In {\em International Workshop on Logic, Language, Information, and
  Computation}, pages 223--236. Springer, 2018.

\bibitem{kidger2022neural}
Patrick Kidger.
\newblock On neural differential equations.
\newblock {\em arXiv preprint arXiv:2202.02435}, 2022.

\bibitem{Ko83}
Ker-I Ko.
\newblock On the computational complexity of ordinary differential equations.
\newblock {\em Information and Control}, 58(1-3):157--194,
  July/August/September 1983.

\bibitem{Ko91}
Ker-I Ko.
\newblock {\em Complexity Theory of Real Functions}.
\newblock Progress in Theoretical Computer Science. Birkha{\"u}ser, Boston,
  1991.

\bibitem{Lorenz63}
E.~N. Lorenz.
\newblock Deterministic nonperiodic flow.
\newblock {\em Journal of the Atmospheric Science}, 20:130--141, 1963.

\bibitem{Miller70}
Webb Miller.
\newblock Recursive function theory and numerical analysis.
\newblock {\em Journal of Computer and System Sciences}, 4(5):465--472, October
  1970.

\bibitem{milnor1985concept}
John Milnor.
\newblock On the concept of attractor.
\newblock {\em Communications in Mathematical Physics}, 99:177--195, 1985.

\bibitem{PoulyGraca16}
Amaury Pouly and Daniel Gra{\c c}a.
\newblock Computational complexity of solving polynomial differential equations
  over unbounded domains.
\newblock {\em Theoretical Computer Science}, 2016.

\bibitem{PouRic79}
Marian~Boykan Pour-El and J.~Ian Richards.
\newblock A computable ordinary differential equation which possesses no
  computable solution.
\newblock {\em Annals of Mathematical Logic}, 17:61--90, 1979.

\bibitem{rojas2023algorithmic}
Cristobal Rojas and Mathieu Sablik.
\newblock On the algorithmic descriptive complexity of attractors in
  topological dynamics.
\newblock {\em arXiv preprint arXiv:2311.15234}, 2023.

\bibitem{Ruo96}
Keijo Ruohonen.
\newblock An effective {C}auchy-{P}eano existence theorem for unique solutions.
\newblock {\em International Journal of Foundations of Computer Science},
  7(2):151--160, 1996.

\bibitem{Sha41}
Claude~E. Shannon.
\newblock Mathematical theory of the differential analyser.
\newblock {\em Journal of Mathematics and Physics MIT}, 20:337--354, 1941.

\bibitem{Sip97}
Michael Sipser.
\newblock {\em Introduction to the Theory of Computation}.
\newblock PWS Publishing Company, 1997.

\bibitem{HolgerThiesPhD}
Holger Thies.
\newblock {\em Uniform computational complexity of ordinary differential
  equations with applications to dynamical systems and exact real arithmetic}.
\newblock PhD thesis, University of Tokyo, Graduate School of Arts and
  Sciences, 2018.

\bibitem{tucker2002rigorous}
Warwick Tucker.
\newblock A rigorous {ODE} solver and {Smale}'s 14th problem.
\newblock {\em Foundations of Computational Mathematics}, 2:53--117, 2002.

\bibitem{Tur36}
Alan {Turing}.
\newblock On computable numbers, with an application to the
  {E}ntscheidungsproblem.
\newblock {\em Proceedings of the London Mathematical Society}, 42(2):230--265,
  1936.
\newblock Reprinted in Martin Davis. The Undecidable: Basic Papers on
  Undecidable Propositions, Unsolvable Problems and Computable Functions. Raven
  Press, 1965.

\bibitem{LivreAnalogcomputing}
Bernd Ulmann.
\newblock {\em Analog computing}.
\newblock Walter de Gruyter, 2013.

\bibitem{ulmann2020analog}
Bernd Ulmann.
\newblock {\em Analog and hybrid computer programming}.
\newblock De Gruyter Oldenbourg, 2020.

\bibitem{veritasum}
Veritasum.
\newblock Future computers will be radically different (analog computing).
\newblock Youtube video, 2022.
\newblock URL: \url{https://www.youtube.com/watch?v=GVsUOuSjvcg}.

\bibitem{Wei00}
Klaus Weihrauch.
\newblock {\em Computable Analysis: an Introduction}.
\newblock Springer, 2000.

\end{thebibliography}

\newpage

\section{Appendix}

\subsection{Very basics of computable analysis} \label{appendix:ca}

In very short, repeating \cite{csl24}, the idea behind classical computability and complexity is to fix some representations of objects (such as graphs, integers, etc.) using finite words over some finite alphabet, say $\Sigma=\{0,1\}$ and to say that such an object is computable when such a representation can be produced using a Turing machine. The computable analysis is designed to be able to also talk about objects such as real numbers, functions over the reals, closed subsets, compacts subsets, \dots, which cannot be represented by finite words over $\Sigma$ (a clear reason for it is that such words are countable while the set $\R$, for example, is not).  However, they can be represented by some infinite words over $\Sigma$ and the idea is to fix such representations for these various objects, called \emph{names}, with suitable computable properties. In particular, in all the following proposed representations, it was proved that an object is computable iff 
it has some computable representation.

\begin{remark}
	Here the notion of computability involved is one of Type 2 Turing machines, that is to say, computability over possibly infinite words: the idea is that such a machine has some read-only input tape(s), that contains the input(s), which can correspond to either a finite or infinite word(s), a read-write working tape and one (or several) write-only output tape(s). It evolves as a classical Turing machine, the only difference being that we consider it outputs an infinite word when it writes forever the symbols of that word on its (or one of its) write-only infinite output tape(s): see \cite{Wei00} for details.
\end{remark}

A name for a point $\vx \in \R^{d}$ is a sequence $(I_{n})$ of nested open rational balls  with $I_{n+1} \subseteq I_n$ for all $n \in \mathbb{N}$ and $\{x\}=\bigcap_{n \in \mathbb{N}} I_n$.   Such a name can be encoded as an infinite sequence of symbols. 

We call a real function $f: \subseteq \mathbb{R} \rightarrow \mathbb{R}$ computable, iff some (Type 2 Turing) machine maps any name of any $x \in \operatorname{dom}(f)$ to a name of $f(x)$. For real functions $\tu f: \subseteq \mathbb{R}^n \rightarrow \mathbb{R}$ we consider machines reading $n$ names in parallel.  A computable function is necessarily continuous: see \cite{Wei00} for all details. 

We also need the concept of polynomial time computable function in computable analysis: see \cite{Ko91}. In short, a quickly converging name of $\vx \in \R^{d}$ is a name  of $\vx$, with $I_{n}$ of radius $<2^{-n}$. 
A function $\tu f: \R^{d} \to \R^{d'}$ is said to be computable in polynomial time, if there is some oracle TM $M$, such that, for all $\vx$, given any fast converging name of $\vx$ as an oracle, given $n$, $M$ produces some open rational ball of radius $<2^{-n}$ containing $\tu f(\vx)$,  in a time polynomial in $n$. 


A function $\tu f: \R^{d} \to \R^{d'}$ is computable in polynomial space if there exists
an oracla TM $M$, such that, for all $\vx$, given any fast converging name of $\vx$ as an oracle, given $n$, $M$ produces some open rational ball of radius $<2^{-n}$ containing $\tu f(\vx)$,  in a space polynomial in $n$. 



\subsection{Some classical statements from numerical analysis}


\begin{lemma}[{Discrete Grönwall's lemma, e.g  \cite[page 213]{Dem96}}] \label{discretegronwall}
  Consider sequences  $h_n,\theta_n \ge 0$ and $\varepsilon_n \in \R$
 such that 
  $$\theta_{n+1} \le (1+\Lambda h_n) \theta_n + |\epsilon_n|$$

  Then
  $$\theta_n \le e^{\Lambda (t_n-t_0) } \theta_0 + \sum_{0 \le i \le
    n-1} e^{\Lambda (t_n-t_{i+1}) } |\epsilon_i|$$
  \end{lemma}

  \begin{proof}
    By recurrence over $n$. For $n=0$, the inequality is $\theta_0
    \le \theta_0$.

    Suppose now the inequality at order $n$. Observe that
    $$ (1+\Lambda h_n) \le e^{\Lambda (t_{n+1}-t_n)}$$

    By hypothesis, we have

    \begin{eqnarray*}
      \theta_{n+1} &\le& e^{\Lambda (t_{n+1}-t_n)} \theta_n +
                         |\epsilon_n| \\
      &\le& e^{\Lambda (t_{n+1}-t_0)} \theta_0 + \sum_{0 \le i \le
            n-1} e^{\Lambda (t_{n+1}-t_{i+1} )} | \epsilon_i| +
            |\epsilon_n|
    \end{eqnarray*}

    The inequality at order $n+1$ follows.

  \end{proof}

\end{document}